\documentclass[11pt]{amsart}

\usepackage{amsmath, amssymb, amsthm, amscd}
\usepackage{url}
\usepackage{graphicx}
\usepackage[hidelinks,pagebackref,pdftex]{hyperref}
\usepackage{color}
\usepackage{import}
\usepackage[letterpaper, margin=1 in]{geometry}
\usepackage{caption}
\usepackage{bbold}
\usepackage{subcaption}
\usepackage{algorithm}
\usepackage{algpseudocode}

\usepackage{amsmath}
\usepackage{old-arrows}
\usepackage{comment}
\usepackage{multirow}
\usepackage{tikz-cd}
\usepackage{enumitem}
\usepackage{stmaryrd}

\newcommand{\Z}{\mathbb{Z}}
\newcommand{\N}{\mathbb{N}}

\newenvironment{localsize}[1]
{%
  \clearpage
  \let\orignewcommand\newcommand
  \let\newcommand\renewcommand
  \makeatletter
  \input{bk#1.clo}%
  \makeatother
  \let\newcommand\orignewcommand
}

\DeclareMathOperator{\Mod}{Mod}
\DeclareMathOperator{\poly}{poly}

\usepackage{marginnote}
\long\def\@savemarbox#1#2{\global\setbox#1\vtop{\hsize\marginparwidth 
  \@parboxrestore\tiny\raggedright #2}}
\renewcommand*{\backref}[1]{}
\renewcommand*{\backrefalt}[4]{
  \ifcase #1
  [No citations.]
  \or [#2]
  \else [#2]
  \fi }

\AtBeginDocument{%
   \def\MR#1{}
}

\numberwithin{equation}{section}
\theoremstyle{plain}
\newtheorem{theorem}[equation]{Theorem}

\newtheorem{lemma}[equation]{Lemma}

\newtheorem*{namedtheorem}{\theoremname}
\newcommand{\theoremname}{testing}

\theoremstyle{definition}

\newtheorem{remark}[equation]{Remark}

\title[Compressed data structures for Heegaard splittings]{Compressed data structures for Heegaard splittings}

\author{Henrique Ennes}\thanks{INRIA d'Universit\'e C\^ote d'Azur, \texttt{henrique.lovisi-ennes@inria.fr}.}
\author{Cl\'ement Maria}\thanks{INRIA d'Universit\'e C\^ote d'Azur, \texttt{clement.maria@inria.fr}.}
\begin{document}

\pagenumbering{gobble}
\begin{abstract}
Heegaard splittings provide a natural representation of closed 3-manifolds by gluing two handlebodies along a common surface. 
These splittings can be equivalently given by two finite sets of meridians lying on the surface, which define a Heegaard diagram.
We present a data structure to effectively represent Heegaard diagrams as normal curves with respect to triangulations of a surface, where the complexity is measured by the space required to express the normal coordinates' vectors in binary.
This structure can be significantly more compact than triangulations of 3-manifolds, yielding exponential gains for certain families.
Even with this succinct definition of complexity, we establish polynomial-time algorithms for comparing and manipulating diagrams, performing stabilizations, detecting trivial stabilizations and reductions, and computing topological invariants of the underlying manifolds, such as their fundamental and homology groups.
We also contrast early implementations of our techniques with standard software programs for 3-manifolds, achieving faster algorithms for the average cases and exponential gains in speed for some particular presentations of the inputs.
\end{abstract}

\maketitle

\pagenumbering{arabic}
\section{Introduction}
Since the early days of computational topology, \emph{3-manifolds} have attracted researchers' keen interest \cite{haken1962homoomorphieproblem, jaco1984algorithm,lackenby2020algorithms3manifoldtheory} for lying exactly midway between the well-understood surfaces and the wildly behaved 4-manifolds.
For example, several 3-manifold problems were related to important questions in complexity theory: $S^3$ recognition was proved by Schleimer to be in $\texttt{NP}$ \cite{schleimer2011sphere}, whereas approximations of some quantum invariants are firmly tied to quantum computing and are known to be $\#\texttt{P}$-hard \cite{alagic2011quantum,alagic2014quantum}.
Similarly, the problems of classifying \cite{costantino2024learning, kuperberg2018identifying} and manipulating the geometry and topology of these spaces \cite{burton2013new, lackenby2024triangulation} have also received significant theoretical and practical attention, with implementations of different methods publicly available in the most important software of the field, \texttt{regina} \cite{regina} and \texttt{SnapPy} \cite{SnapPy}. 

The algorithmic point of view is particularly relevant for the construction of \emph{censuses} of homomorphism classes of 3-manifolds \cite{burton2004efficient, burton2014cusped, Burton2020_TabulationKnots3Manifolds, burton2015edge}. 
Drawing inspiration from the related problem of knot classification (this time, up to ambient isotopies), where the cumulative work of researchers \cite{burton2020next, conway1970enumeration, dowker1983classification} has culminated in the widely accessible census of $\sim 350'000'000$ prime knot diagrams, computational methods are expected to foster the elaboration of tables of 3-manifolds, particularly valuable assets in the proposition and verification conjectures, as well as in the discovery of counterexamples. 

Abstracting modern approaches \cite{burke_et_al:LIPIcs.SoCG.2025.28, burton2014cusped, burton2020next, burton2015edge} of the construction of censuses of topological spaces, one requires 1. compact data structures to represent massive amounts of (candidate) spaces and the iterative construction of all spaces up to a maximum complexity; 2. efficient operations to disqualify inputs (such as {composite} knots, links, non-prime 3-manifolds) and simplify the representation combinatorics (Reidemeister moves, Pachner moves, etc); and, 3. methods for discriminating spaces, such as invariants. 
For 3-manifolds, these steps are often addressed through \emph{triangulations}---homeomorphic copies described by a finite set of tetrahedra and gluing rules among their faces---of complexity usually measured by the number of tetrahedra. 
On the other hand, theoretical bounds on the minimal number of tetrahedra to represent some spaces \cite{datta2007minimal, jaco20030, jaco2006layered, matveev2007algorithmic} and the difficulty of generating particular structures, such as hyperbolicity \cite{matveev2007algorithmic}, can limit the effectiveness of triangulation-based censuses.

\emph{Heegaard splittings} emerge as an alternative to triangulations.
They depict a closed 3-manifold $M$ as the result of gluing two better-behaved pieces called \emph{handlebodies} along a common boundary. 
Splittings can also be combinatorially described by \emph{Heegaard diagrams}, two sets of curves in a surface whose intersection pattern conveys much of the topological information about $M$.
We suggest here a new data structure to represent Heegaard diagrams based on \emph{normal coordinates}, of a size as large as the \emph{binary} of the number of intersections between the two sets of curves. 
The gain in measuring the complexity of 3-manifolds through this notion is not to be overlooked: we present in Section \ref{sec: examples and experiments} families of manifolds whose diagrams' complexity grows linearly with the size of the input, but whose triangulations necessarily use exponentially many tetrahedra. 
However, what is more important is that due to the proliferation of efficient operations in normal coordinates \cite{Erickson2013, lackenby2024some, schaefer2002algorithms, schaefer2008computing, Stefankovic}, we can establish polynomial-time algorithms even in this compressed notion of complexity, leading to basic manipulations of Heegaard splittings and the computation of simple 3-manifold invariants such as the fundamental and homology groups. 
Moreover, as we point out in Remark \ref{rm: co-NP}, some important complexity results about 3-manifolds---such as sphere (co-)recognition---are still preserved if we assume compressed inputs, which tells us much about the inherent difficulty of these problems.
We hope that these techniques are useful for the study and classification of topological spaces that are beyond the reach of efficiency when represented by other methods.

The paper begins with a background section on both low-dimensional and computational topology (Section \ref{sec: background}), where we review some of the concepts and results that will be needed in the later parts. 
We refer the reader to \cite{schultens2014introduction} and \cite{singer2015lecture} for basic ideas of geometric topology, such as homeomorphisms, isotopies, gluing maps, and manifolds.
Our contribution starts in Section \ref{sec: data structure}, where we formally present our data structure for Heegaard splittings. 
Section \ref{sec: algorithms} establishes basic algorithms in this structure.
We compare early implementations\footnote{The code is publicly available as the \texttt{FHDPy} package at \url{https://github.com/HLovisiEnnes/FHDpy}. Details of the implementation are also outlined in Appendix \ref{app: details}.} of some of these algorithms with triangulation-dependent methods in Section \ref{sec: examples and experiments}.
Finally, Section \ref{sec: open questions} points to open questions towards the construction of censuses of 3-manifolds.
We also present appendices with extensions of the techniques described in the main text and details on the implementation.

Before we close this introduction, we draw the reader's attention to two subtleties of the current work. 
First, although most of our tools apply to \emph{generalized} Heegaard splittings defined for \emph{orientable compact} 3-manifolds, potentially with boundary \cite{johnson2006notes, schultens2014introduction}, we will focus solely on \emph{closed} 3-manifolds. 
Therefore, all 3-manifolds here will be assumed orientable and closed, unless explicitly stated otherwise. 
Second, we assume throughout a \textit{unit-cost RAM model} with bit-size $W$ \cite{hagerup1998sorting}. 
This means that the four integer operations and integer comparison can be executed in constant time, provided that the numbers involved have binary representations of size at most $W$.

\section*{Acknowledgement} This work has been partially supported by the ANR project ANR-20-CE48-0007 (AlgoKnot) and the project ANR-15-IDEX-0001 (UCA JEDI). It has also been supported by the French government, through the France 2030 investment plan managed by the Agence Nationale de la Recherche, as part of the ``UCA DS4H'' project, reference ANR-17-EURE-0004.

\section{Background material}\label{sec: background} 
\subsection{Background on low-dimensional topology}\label{sec: background topology}
By \textit{surface}, we understand a compact and connected smooth 2-manifold $\Sigma$ with or without boundary.
All surfaces will be assumed orientable and oriented. 
If $\Sigma$ is closed, it is uniquely given, up to homeomorphism, by an integer $g$ called the surface's genus \cite{farb2011primer} and we denote this surface by $\Sigma_g$. 
A \textit{(simple closed) curve} and an \emph{arc} are images of proper embeddings of $S^1$ and $[0,1]$ in $\Sigma$, respectively. 
A \textit{multicurve} is a finite collection of disjoint curves. 
Given a multicurve $\gamma$, we denote its number of connected components by $\#\gamma$ and, with some abuse of notation, write $c\in \gamma$ if $c$ is a component of $\gamma$. 
Whenever possible, we reserve Greek letters for multicurves and Latin letters for curves and arcs.

We will most often consider the isotopy class of a (multi)curve $\gamma$ in $\Sigma$, $[\gamma]$.
A class $[\gamma]$ is called \emph{essential} if no component is isotopic to a point.
From now on, we will not differentiate $\gamma$ from its isotopy class $[\gamma]$ in the notation if there is no confusion.

 \begin{figure}
    \centering
    \includegraphics[width=0.5\textwidth]{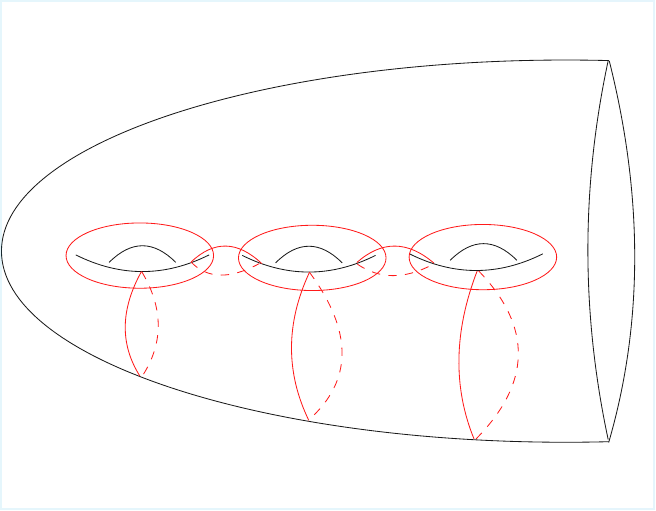}
     \caption{The recursive relation that defines the Lickorish generators of the mapping class group.}
  \label{fig: lockorish}
\end{figure}   

Let $\gamma$ and $\delta$ be two oriented (multi)curves in $\Sigma$, isotoped to intersect transversely.
For each intersection point, we let its \emph{index} be $+1$ if the orientation of the curves agrees with the orientation of $\Sigma$ as in the right-hand rule, otherwise we set it to $-1$.
The \textit{algebraic intersection number} of $\gamma$ and $\delta$, $\hat{i}(\gamma, \delta)$, is the sum of the indices of all their transverse intersections; the algebraic intersection number is well defined up to isotopy \cite{farb2011primer}.
Two (multi)curves are said to be in \textit{efficient position} if they are representative of their respective isotopy classes such that the total number of intersection points between them is minimal.
Equivalently, they are in efficient position if they do not have segments that cobound disks, known as \textit{bigons}, on the surface \cite{farb2011primer}. 

The \emph{mapping class group} of the surface $\Sigma_g$, $\Mod(\Sigma_g)$, is the group of orientation-preserving $\Sigma_g\to\Sigma_g$ homeomorphisms up to isotopy. 
For each $g$, $\Mod(\Sigma_g)$ is finitely generated as a group by an alphabet $\mathcal{L}$, consisting of \emph{Dehn twists} about $3g-1$ (isotopy classes of) canonically defined essential curves $s_1,\dots,s_{3g-1}$ \cite{lickorish1964finite} (Figure \ref{fig: lockorish}).
Formally, a Dehn twist $\tau_s:\Sigma_g\to\Sigma_g$ about $s$ is given by cutting out a cylindrical regular neighborhood around the curve, giving a $2\pi$ left twist to one of the cylinder's boundaries and gluing the cylinder back into the surface.
We call the curves $s_1,\dots,s_{3g-1}$ the surface's \textit{Lickorish generators}. 
Hopefully without confusion, we also call the Dehn twists in the alphabet $\mathcal{L}=\{\tau_{s_1},\dots,\tau_{s_{3g-1}}\}$ Lickorish generators.

An essential multicurve $\gamma$ in the surface $\Sigma_g$ is a \textit{system} if $\Sigma_g$ cut along $\gamma$, $\Sigma_g\backslash \gamma$, is homeomorphic to a $2g$ punctured sphere. 
A system $\gamma$ in $\Sigma_g$, for $g\geq 1$, can be used to construct the genus $g$ \textit{handlebody}
\begin{equation*}
    \mathcal{H}_\gamma = \Sigma_g\times [0,1]\cup_{\gamma \times \{0\}} 2\text{-handles}\cup 3 \text{-handles}
\end{equation*}
where we attach 2-handles along the curves $\gamma$ in $\Sigma\times \{0\}$ and fill any spherical boundary component with 3-handles. 
Note that $\partial\mathcal{H}_\gamma=\Sigma_g$ and that the handlebody $\mathcal{H}_\gamma$ is homeomorphic to the connected sum of $g$ solid tori.
The components of $\gamma$ are \textit{meridians} of $\mathcal{H}_\gamma$ in the sense that they bound disks in the handlebody's interior. 
Two systems $\gamma$ and $\gamma'$ in the same surface $\Sigma_g$ will be called \textit{equivalent} if $\gamma'$ is a set of meridians in $\mathcal{H}_\gamma$ (which implies that $\gamma$ is a set of meridians in $\mathcal{H}_{\gamma'}$).

\begin{remark}\label{rm: systems}
    In this paper, we restrict the analysis to \emph{minimal} systems $\gamma$, that is, we assume $\#\gamma=g$.
    In recent decades, it became common to extend the definition of a system to any multicurve in the surface that defines a handlebody \cite{hempel20013, yoshizawa2014high}.
    In fact, the data structure described in Section \ref{sec: data structure} can be generalized to this notion of system, and there are algorithms to translate from and to the special case of cardinality $g$.
    For completeness, we describe these extensions in the Appendix \ref{app: extensions}.
\end{remark}

Let $c_1,c_2$ be two fixed components of a system $\gamma$ and $p$ an arc from a puncture corresponding to $c_1$ to a puncture corresponding to $c_2$ in $\Sigma_g\backslash \gamma$. 
We define the \emph{band sum} of $c_1,c_2$ and $p$ as the curve $c$ in $\Sigma_g\backslash \gamma$ indicated by the surgery in Figure \ref{fig: band sum}. 
One can show that $c$ is essential and bounds a disk in $\mathcal{H}_\gamma$ \cite{ennes2025hardnesscomputationquantuminvariants, johnson2006notes}, implying that $\gamma'=\gamma-\{c_1\}\cup \{c\}$ is a system equivalent to $\gamma$. 
We call this whole operation a \emph{disk slide}.
With finitely many applications of disk slides and isotopies, a system $\beta$ can be transformed into any equivalent system $\beta'$ \cite{johnson2006notes}. 

\begin{figure}
    \centering
    \includegraphics[width=0.6\textwidth]{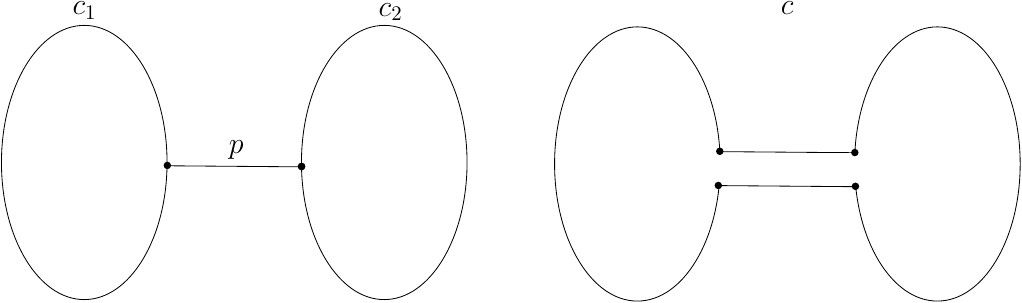}
     \caption{The result $c$ of a band sum of $c_1$, $c_2$, and $p$.}
  \label{fig: band sum}
\end{figure}

A genus $g$ \emph{Heegaard splitting} is the closed 3-manifold resulting from the gluing of two handlebodies with boundaries homeomorphic to $\Sigma_g$ through a gluing map $f:\Sigma_g\to\Sigma_g$.
In fact, one might assume that $f$ is given by a $\phi\in \Mod(\Sigma_g)$ composed with some canonical orientation-reversing homeomorphism. 
This means that a Heegaard splitting can be given by a pair $(\Sigma_g,\phi)$ where $\phi$ is a word on Lickorish generators.
We call this presentation of Heegaard splitting a \emph{Heegaard word}.
Every closed 3-manifold has a Heegaard splitting and can then be associated with some pair $(\Sigma_g,\phi)$ \cite{moise1952affine, saveliev2011lectures}.

Although naturally combinatorial, a Heegaard word $(\Sigma_g,\phi)$ hardly conveys topological information about the closed 3-manifold $M$ it represents.
Nonetheless, if we fix a system of meridians $\alpha$ on one of the two handlebodies, $\mathcal{H}_\alpha$, then $\beta=\phi(\alpha)$ is a system of meridians of the \emph{other} handlebody, $\mathcal{H}_\beta$. 
We call the tuple $(\Sigma_g, \alpha, \beta)$ a \textit{Heegaard diagram}.
We note that Heegaard diagrams are algorithmically independent of Heegaard words: any two systems $\alpha$ and $\beta$ will define the same splitting, provided that the curves are meridians in their respective handlebodies. 
Diagrams are therefore not unique for a given splitting. 
For example, if $\beta$ and $\beta'$ are \textit{isotopic diagrams} on $\Sigma_g$---that is, $\beta$ and $\beta'$ are isotopic multicurves in $\Sigma_g$---then $(\Sigma_g,\alpha,\beta)$ and $(\Sigma_g,\alpha,\beta')$ naturally define the same manifold.
Similarly, if $\beta$ is equivalent to $\beta'$, we say that the diagrams are \textit{equivalent} and note that $(\Sigma_g,\alpha,\beta)$ and $(\Sigma_g,\alpha,\beta')$ also define the same 3-manifold.

\emph{Stabilizations} provide yet another method to generate alternative Heegaard splittings of the same 3-manifold. 
Formally, let $(\Sigma_{g_1},\alpha_1,\beta_1)$ be a Heegaard diagram of a 3-manifold $M_1$ and $(\Sigma_{g_2},\alpha_2,\beta_2)$ be a diagram of $M_2$. 
The diagram $(\Sigma_{g_1+g_2},\alpha_1\cup\alpha_2,\beta_1\cup \beta_2)$ represents the connected sum $M_1\#M_2$, where $\Sigma_{g_1+g_2}$ is the connected sum of $\Sigma_{g_1}$ and $\Sigma_{g_2}$ taken along open disks disjoint from $\alpha_1$,$\alpha_2$, $\beta_1$, and $\beta_2$.
In particular, noting that $S^3$ has a splitting with diagram $(\Sigma_1=T^2,m,\ell)$, where $|m\cap \ell|=1$, and that $M\#S^3$ is homemorphic to $M$, the splitting $(\Sigma_{g+1},\beta\cup\{\ell\})$ represents the same manifold as $(\Sigma_g,\alpha,\beta)$.
We say that the former is a {stabilization} of the latter.
In the other direction, a Heegaard diagram $(\Sigma_g,\alpha,\beta)$ represents a stabilization of a lower genus splitting if and only if there is a meridian $a$ in $\mathcal{H}_\alpha$ and a meridian $b$ in $\mathcal{H}_\beta$ such that $|a\cap b|=1$ \cite{johnson2006notes}. 
In particular, we say that this pair of meridians is a \textit{stabilization pair} and that the diagram is \textit{stabilized}. 
For a Heegaard diagram $(\Sigma_g,\alpha,\beta)$, $a$ and $b$ form a \textit{trivial stabilization pair} if $a$ is a component $ \alpha$, $b$ is a component of $\beta$, they form a stabilization pair, and $b\cap a'=\emptyset$ and $b'\cap a = \emptyset$ for any other component $a'\in\alpha$ and $b'\in\beta$. 
Whenever $a$ and $b$ form a trivial stabilization, cutting the surface along $a$ induces a Heegaard diagram in the surface of genus one unit smaller.
This process will be called \textit{destabilization}.
 
A splitting $(\Sigma_g,\mathcal{H}_\alpha,\mathcal{H}_\beta)$ is called \textit{reducible} if there is a curve $c$ in $\Sigma_g$ that is a meridian of both $\mathcal{H}_\alpha$ and $\mathcal{H}_\beta$. 
We call a diagram $(\Sigma_g, \alpha,\beta)$ \textit{reducible} if either (1) there are a component $a\in \alpha$ and a component $b\in \beta$ that are isotopic in $\Sigma_g$; or (2) if there is an essential separating curve $s$, disjoint from both $\alpha$ and $\beta$, and such that $\Sigma_g\backslash s$ defines a connected sum of two Heegaard splittings. 
The uniqueness of fundamental group decomposition \cite{aschenbrenner20153} and \cite{schultens2014introduction}[Theorem 6.3.5] implies that a reducible diagram represents a reducible splitting.
A 3-manifold is \emph{reducible} if it contains a sphere that does not bound a ball. 
If $(\Sigma_g,\alpha,\beta)$ is stabilized, then either the splitting is reducible or a stabilization of a genus 1 splitting of $S^3$; moreover, a reducible splitting is stabilized or a splitting of a reducible manifold \cite{johnson2006notes}.

Useful algebraic information on the 3-manifold can also be readily obtained from its Heegaard diagrams.
Label the curves $a_1,\dots,a_g\in\alpha$ and $b_1,\dots,b_g\in \beta$ and give each of them an orientation. 
It follows from an application of the Van Kampen theorem \cite{hatcher2002algebraic, hempel20223, johnson2006notes} that we can compute a presentation $\langle x_1,\dots,x_g|r_1,\dots, r_g\rangle$ for the \emph{fundamental group} $\pi_1(M)$ by traversing the curve $b_i$ along the assumed orientation and, whenever it intersects some $a_j$, we append $x_j$ to the relation $r_i$ if the index of the intersection is positive, otherwise we append $x_j^{-1}$ to it (if $b_i$ does not intersect $\alpha$, $r_i$ is just the trivial relation). 

Because the \textit{first (integral) homology group of} $M$, $H_1(M)$, is the abelianization of $\pi_1(M)$ \cite{hatcher2002algebraic},
we can present $H_1(M)$ as $ \langle x_1,\dots,x_g| r'_1,\dots, r'_g, [x_i,x_j]\rangle$, where each $r'_i= x_1^{k_i^1}\dots x_g^{k_i^g}$ can be computed by commuting the generators $x_j$ in the relations $r_i$.
Note that $k_i^j=\hat{i}(b_i, a_j)$ and we store this information in a $g\times g$ \textit{presentation matrix} for $H_1(M)$, $K=[k_i^j]$, whose Smith normal form carries all the necessary information about this Abelian group. 
For example, the \textit{(first) Betti number} of $M$ is given by the number of zero-rows. 
A 3-manifold with Betti number equal to zero is called a \textit{rational homology sphere}, and it will be an \textit{integral homology sphere} if the Smith form of the presentation matrix is the identity.

\subsection{Background on computational topology}
Heegaard diagrams consist of three essential components that any candidate data structure must encode: a surface, $\Sigma_g$, and two multicurves, $\alpha$ and $\beta$.
Starting with the surface, a graph $T=(V,E)$ embedded in $\Sigma$ is a \textit{cellular embedding} if the manifold defined by cutting $\Sigma$ along $T$, $\Sigma\backslash T$, is homeomorphic to a collection of disks. 
We call each disk component of $\Sigma\backslash T$ a \textit{face} of $T$.
A cellular embedding $T$ in the surface $\Sigma$ is said to be \textit{oriented} if we assign to each face an ordered list of its vertices, consistent with the orientation of $\Sigma$ as in the right-hand rule. 
We assume every cellular embedding to be oriented.
A cellular embedding $T$ is a \textit{(generalized) triangulation} of $\Sigma$ if exactly three edges of $T$ bound (the closure of) each face. 
We define the size of a triangulation, $|T|$, as its number of faces; note that $|V|,|E|=O(|T|)$ where $|V|$ is the number of vertices of the triangulation and $|E|$ the number of edges. 
Surface triangulations should not be confused with the similar, but different concept of 3-manifold triangulations.

Every oriented surface can be described as a data structure by an oriented triangulation.
For the triangulations themselves, we will assume that they are given by the set of vertices, edges, and triangles, with bidirectional pointers between vertices and edges, and between edges and the triangles, which therefore requires a total space of $O(|T|)$.
In particular, determining whether two triangles are adjacent, getting the incident vertices and edges of a triangle, and finding triangles bounded by a specific edge take constant time in this model. 
Other, more compact data structures for triangulation, such as CGAL's \cite{devillers2000triangulations}, can be converted into our structure in time $O(|T|)$ and at similar costs to memory.
More general cellular embeddings are treated similarly.

A curve $c$ is \textit{edged} with respect to $T$ if it is disjoint from its faces (red curve in Figure \ref{fig: data structure}).
We represent an edged curve as a data structure by the list $E_T(c)=\{e\in E:|c\cap e|=e\}$, assumed to be ordered according to an implicit traversing orientation of $c$. 
Similarly, a multicurve $\gamma$ can be represented by a set of $\#\gamma$ edge lists, one for each component, which we denote as $E_T(\gamma)$. 
We assume that no two components can share an edge, nor can the same component be incident to a vertex more than once; nevertheless, two distinct components can share a vertex. 
The \textit{complexity} of $E_T(\gamma)$ is $\|E_T(\gamma)\|=\sum_{c\in \gamma}|E_T(c)|$.

\begin{remark}
    The results of this paper, including Lemma \ref{lm: get twins}, would still hold---sometimes with an extra parameter in the complexity---even if we allowed distinct components of edge-represented multicruves to share edges.
\end{remark}

On the other hand, an essential multicurve $\gamma$ is called \textit{standard} in $T$ if it intersects the graph only at the edges and transversely (blue curve in Figure \ref{fig: data structure}).
If $\gamma$ is standard, we can represent it as a collection of words in a finite alphabet: we arbitrarily orient the edges $E$ and take $E^{\pm} = E\cup E^{-1}$ for the letters.
Then for each component $c$ of $\gamma$, the word $I_T(c)$ in $(E^{\pm})^*$ is given by traversing $c$ along some arbitrary direction and, whenever we intersect an edge $e\in E$, we append $e^i$ to $I_T(c)$ where $i$ is the index of the intersection. 
An \textit{intersection sequence} representation for $\gamma$, $I_T(\gamma)$, is the set $\{I_T(c)\}_{c \in \gamma}$.
While $E_T(\gamma)$ is well-defined at the multicurve level, intersection sequences are defined only up to isotopies inside the faces of the cellular embedding.
Moreover, cyclic permutations and taking inverses of the words $I_T(c)$ in $I_T(\gamma)$ yield the same isotopy classes.

Instead of working with explicit intersection sequences, we will often assume \textit{straight line programs} (\emph{SLPs}) to express $I_T(c)$.
Formally, an SLP on the alphabet $E^{\pm}$ is a finite sequence of equations $\langle x_1 = \text{EXPR}_1, \dots, x_\ell = \text{EXPR}_\ell\rangle$, where the left-hand side, $x_i$, are symbols not in $E^{\pm}$ called \textit{variables}, and the expressions on the right-hand side, $\text{EXPR}_i\in(E^{\pm}\cup\{x_1,\dots,x_{i-1}\})^*$, are called \textit{assignments}.
Each assignment $\text{EXPR}_i$ is assumed to either be a symbol from $E^{\pm}$, in which case the assignment is called \textit{simple}, or of the form $\text{EXPR}_i=x_{j_1}^{\pm 1}\cdot x_{j_2}^{\pm 1}\cdot \dots \cdot x_{j_k}^{\pm 1}$, where $1\leq j_1,\dots,j_k < i$, in which case we say the assignment is \textit{proper}.
In particular, the SLP $\langle x_1,\dots,x_\ell\rangle$ represents a word $w$ in $(E^{\pm})^*$ if, recursively substituting the appropriate $x_{j}$ in each proper assignment, gives $x_\ell$ equal to $w$.
We call $\ell$ the \textit{length} of the SLP and define its \textit{complexity} $m=\sum_{i=1}^\ell |\text{EXPR}_i|$, where $|\text{EXPR}_i|$ is the number of variables and symbols in the assignment $|\text{EXPR}_i|$.
Given any word $e_1\dots e_n$ in $(E^{\pm})^*$, one can trivially build an SLP of complexity $2n$, but there are instances in which SLPs can be used to give exponential gains over uncompressed representation of words \cite{schaefer2008computing}. 
Nonetheless, even in these extremely compressed cases, several useful operations can still be efficiently performed with respect to the SLP's complexity, as stated in the following lemma. 
Here and throughout, we assume that $\lg k$ is the number of bits necessary to represent the potentially signed integer $k$; in particular, $\lg k = \log k + 2$ if $k$ is signed and $\log k +1$ if unsigned.

\begin{lemma}\label{lm: slp}
    Let $\langle x_1,\dots,x_\ell\rangle$ be an SLP of complexity $m$ representing a word $w\in (E^{\pm})^*$. In a unit-cost RAM model of bit-size $m + \lg k$, for a fixed $k\in \Z^+$, one can compute
    \begin{enumerate}[label=(\alph*)]
        \item an SLP of $w^{-1}$ in time $O(1)$;
        \item an SLP of $w^k$ in time $O(\lg k)$;
        \item the number of occurrences of a symbol $e\in E^{\pm}$ in $w$, $e(w)$, in time $O(m)$.
    \end{enumerate}
\end{lemma}

\begin{proof}
    For (a), note that $\langle x_1,\dots, x_\ell,x_\ell^{-1}\rangle$ is an SLP for $w^{-1}$. 
    For (b), assume $k\geq 0$ (the case for $k<0$ follows from this and (a)) and let $b_0b_1\dots b_{\lg k}$ be the binary expansion of $k$, then 
    \begin{equation*}
        \langle x_1,\dots, x_\ell, x_{\ell+1}=x_\ell\cdot x_\ell, \dots, x_{\ell+\lg k}=x_{\ell+\lg k-1} \cdot x_{\ell+\lg k-1}, b_0 x_\ell \cdot b_{1}  x_{\ell+1}\cdot \dots \cdot x_{\ell+\lg k}\rangle
    \end{equation*}
    is an SLP for $w^k$. 
    Finally, (c) follows by dynamic programming through iterating over $1\leq i \leq \ell$ and counting the references to the simple assignment $x_j=e$ or to proper assignments that recursively refer to assignments that refer to $x_j$. 
\end{proof}

We will assume that every intersection sequence $I_T(c)$ is given by an SLP of complexity $\|I_T(c)\|$ and let $\|I_T(\gamma)\|=\sum_{c\in \gamma}\|I_T(c)\|$.

\begin{figure}
\begin{subfigure}{.48\textwidth}
  \centering
  \includegraphics[width=.6\linewidth]{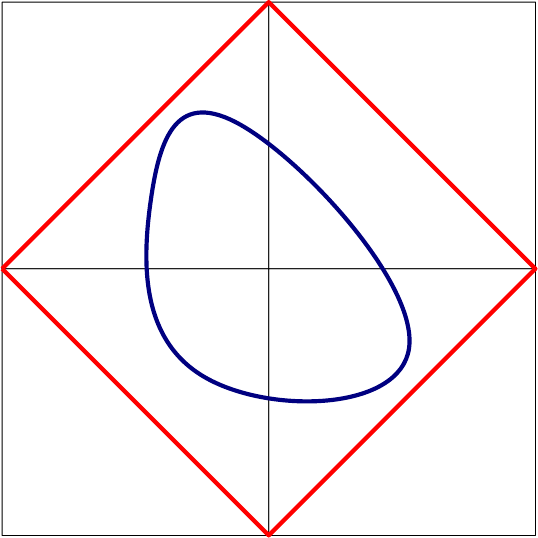}
  \caption{A triangulation with an edged (red) and a standard (blue) curve.}
         \label{fig: data structure}
\end{subfigure}\hfill
\begin{subfigure}{.48\textwidth}
  \centering
  \includegraphics[width=.6\linewidth]{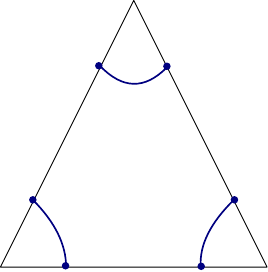}
         \caption{The three distinct isotopy classes of elementary arcs (blue) on a triangular face.}
         \label{fig: elementary arcs}
     \end{subfigure}
        \caption{}
\end{figure}

Now, suppose that $T=(V,E)$ is a triangulation. 
A standard (multi)curve $\gamma$ is \textit{normal} if it intersects each face of $T$ in arcs, called \textit{elementary arcs}, that connect distinct sides of the triangle in which they are contained (Figure \ref{fig: elementary arcs}). 
Equivalently, $\gamma$ is normal if each (uncompressed) intersection sequence in $I_T(\gamma)$ is \textit{cyclic reduced}, that is, if there are no substrings of type $ee^{-1}$ or $e^{-1}e$ in any of the cyclic permutations of the components' sequences. 
When the curve $c$ is normal for a triangulation, its isotopy class is fixed by the number of times it intersects each labeled edge \cite{schaefer2002algorithms}, which implies that we can describe it through a vector in $N_T(c)$ in $\N^{|E|}$. 
We call $N_T(c)$ the \textit{normal coordinates} of the curves, whose complexity, $\|N_T(c)\|$, is the number of bits necessary to store this vector, i.e., $\sum_{e\in E}\lg e(c)=\sum_{e\in E}(\log e(c)+1)$.
Naturally, we define $N_T(\gamma)$ as the set $\{N_T(c)\}_{c\in \gamma}$.

\begin{remark}
    We could equally represent a normal multicurve $\gamma$ by a \emph{single} vector $\sum_{c\in\gamma} N_T(\gamma)$. 
    Similarly, it is possible (and common) to extend the notion of SLPs to encode more than one uncompressed word, which, in our case, implies giving all components of $\gamma$ as a single SLP. 
    Although these approaches are more parsimonious, they are less useful for our needs.
    Moreover, Lemma 5.1 of \cite{Erickson2013} gives an algorithm to transform, in time $O(|T|\|N_T(\gamma)\|)$, the normal coordinates of a multicurve given by a sum of its components' coordinates to our expected form.
\end{remark}

Any normal (multi)curve $\gamma$ partitions the edges of $T$ into segments called \emph{ports}. 
The \emph{overlay graph} of $\gamma$ is defined by having as vertices the union of $V$ and the intersections $E\cap \gamma$, and as edges the union of the ports and the elementary arcs of $\gamma$ (see Figure \ref{fig: street complex}). 
The overlay graph splits each triangle $t$ of the original triangulation into two types of regions: \emph{junctions}, which are adjacent to three edges of $T$, and \emph{blocks}, which are adjacent to only two. 
A port of the overlay graph is called \emph{redundant} if it separates blocks from different triangles of the original triangulation. 
Following Erickson and Nayyeri, we define the \emph{street complex} of $\gamma$, $S(T,\gamma)$, by deleting redundant ports and their endpoints from the overlay graph.
In practice, the street complex merges blocks of the overlay graph into longer faces called \emph{streets}; naturally, $S(T,\gamma)$ is a cellular embedding in the surface triangulated by $T$, see Figure \ref{fig: street complex}.
We define the \emph{crossing sequence} of a street in $S(T,\gamma)$ as the intersection sequence of a curve that crosses the street from one end to the other with respect to the original triangulation $T$. 
For simplicity, we assume that any port that separates two junctions is a degenerate street with a crossing sequence of length 1. 
The street complexes of all multicurves in this paper will have complexity $O(|T|)$ \cite{Erickson2013}. 

\begin{figure}
    \centering
    \includegraphics[width=0.8\textwidth]{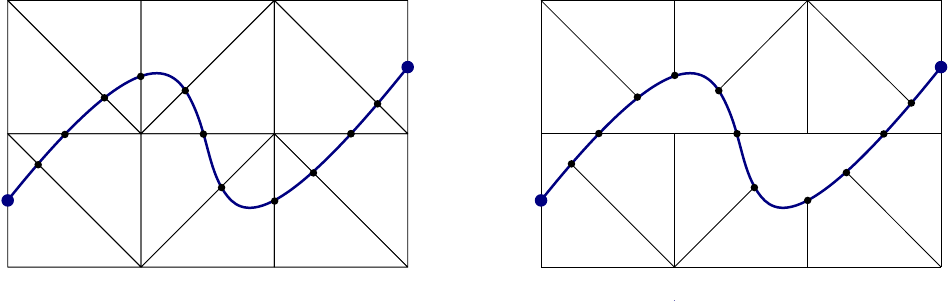}
     \caption{Left: the overlay graph of a normal arc (blue) with respect to a triangulation (black). Right: its street complex.}
  \label{fig: street complex}
\end{figure}

\section{Data structure for Heegaard diagrams}\label{sec: data structure}
Normal coordinates suggest a very natural encoding of Heegaard diagrams and, consequently, of closed 3-manifolds.
Explicitly, we propose representing $(\Sigma_g, \alpha, \beta)$ by a tuple $(T,E_T(\alpha), N_T(\beta))$, where $T=(V,E)$ is a triangulation, $E_T(\alpha)$ is an edge representation of $\alpha$, and $N_T(\beta)$ is a normal coordinate representation of $\beta$. 
We will often refer to the pair $(T,E_T(\alpha))$ as a \textit{marked triangulation} of $\Sigma_g$, in which context $N_T(\beta)$ will be called the diagram itself.
The \textit{complexity} of the diagram is
\begin{equation}\label{eq: complexity}
    \|(T,E_T(\alpha),N_T(\beta))\|=|T|+\|E_T(\alpha)\|+\|N_T(\beta)\|.
\end{equation}
We note that $\|E_T(\alpha)\| = O(|T|)=O(\|N_T(\beta)\|/|\beta\cap E|)$, that is, for complexity purposes, only the last term of \eqref{eq: complexity} matters.
However, for the sake of clarity, we will still refer to the terms of equation \eqref{eq: complexity} arising from the marked triangulation in our analyses.
All complexities will be reported under the assumption that $|\beta\cap T|=\Omega(|T|^2)$, as otherwise there are some faster trivial algorithms.

Our choices of data structures for $\alpha$ and $\beta$ allow us to represent curves with exponentially many intersection points with respect to the complexity of the diagram.
Referring to the discussion at the end of Section \ref{sec: background topology}, more intersections may lead to more complicated fundamental groups and, consequently \cite{matveev2007algorithmic}[Proposition 2.6.6], more complicated 3-manifolds.
In this sense, our structure is significantly more {compressed} than representations of Heegaard diagrams with \emph{both} $\alpha$ and $\beta$ edged, which can have, at most, $E_T(\alpha)\times E_T(\beta)=O(|T|^2)$ intersections. 

Complexity gains are also noticeable if we compare our structure with triangulations of 3-manifolds. 
The usual algorithm to triangulate Heegaard diagrams \cite{he2023algorithm} takes polynomial-time in $|\alpha\cap \beta|$. 
Moreover, as we empirically verify in Section \ref{sec: examples and experiments}, converting a Heegaard word to a diagram can be exponentially faster than converting the word to a triangulation for some choices of input and at least as fast for the general case. 
On the downside, it is unclear how to obtain and manipulate geometric information (such as \emph{hyperbolicity}) or to compute certain quantum invariants directly from Heegaard diagrams, making triangulations still preferable for these kinds of problems, but more on this is discussed in Section \ref{sec: open questions}.

As stated in the introduction, our goal is to establish basic efficient algorithms even in this very compressed setting. For such a purpose, we will make extensive use of the next two results. 
The first of them is (almost) fully established by Erickson and Nayyeri \cite{erickson2002optimally}.

\begin{theorem}\label{th: erickson}
        Suppose that $(T,E_T(\alpha),N_T(\beta))$ is a genus $g$ Heegaard diagram of complexity $m$. In a unit-cost RAM model of bit-size $\log |\beta\cap E|=O(m/|T|)$, one may compute 
        \begin{enumerate}[label=(\alph*)]
            \item the street complex $S(T,\beta)$ in time $O(m|T|)$;
            \item SLPs of the crossing sequence of all streets of $S(T,\gamma)$ in time $O(m|T|)$ of \emph{total} complexity $O(m|T|)$;
             \item the SLPs $I_T(\beta)$ of complexity $O(m)$, in time $O(m)$.
        \end{enumerate}
\end{theorem}

 \begin{proof}
     Part (a) corresponds to Theorems 4.7 of \cite{Erickson2013}. 
     Part (b) follows from Theorem 4.8 of the same paper, where we simply add the variable corresponding to the desired street at the end of the SLP compressing the crossing sequences of all components of $\beta$. 
     For Part (c), use Lemma 3.4.2 of~\cite{Stefankovic}.
    \end{proof}
\begin{remark}
    A conceptually simpler algorithm for part (c) of Theorem \ref{th: erickson} using only the methods of \cite{Erickson2013} is possible, but it outputs SLPs of complexity $O(m|T|)$.
\end{remark}

For the next theorem, we might assume that a Heegaard word $\phi$ is given in \textit{power-notation form}, that is, $\phi=\tau^{k_n}_{s_n}\circ \dots \circ \tau^{k_1}_{s_1}$, $s_i\neq s_{i+1}$ for $1\leq i< n$, and with each $k_i\in \Z$ in signed binary. 
A \textit{fully uncompressed} Heegaard word $\phi'=\tau^{\pm 1}_{s'_{|\phi'|}}\circ\dots \tau^{\pm 1}_{s'_1}$, with some $s'_i$ potentially equal to $s'_{i+1}$, can be put in power-notation in time $O(|\phi'|)$ in a unit-cost RAM model of bit-size $\lg |\phi'|$. 
For convenience, we shall also assume that the Lickorish generators are given as edge curves with respect to $T$.
Nevertheless, even if we take them given by normal coordinates or intersection sequences, the theorem still holds with similar complexity bounds by multiple applications of Theorem \ref{th: do a dehn twist}. 
Before we move to the main result, we need the following (technical) lemma.
\begin{lemma}\label{lm: get twins}
    Let $E_T(\gamma)$ be an edge list representation of an essential multicurve $\gamma$ in the surface $\Sigma_g$, with triangulation $T$. 
    One can compute, in time $O(\#\gamma \times |T|)$, a set of SLPs of complexity at most $O(\#\gamma \times |T|)$ of the intersection sequences of a normal multicurve $\gamma'$, isotopic to $\gamma$ in $\Sigma_g$.
\end{lemma}
\begin{proof}
Suppose $\gamma$ is a simple curve. 
We can obtain, in time $O(|T|)$, the first barycentric subdivision $T'$ of the triangulation $T$ \cite{hatcher2002algebraic}, which we note to have size $|T'|=6|T|$ (see the blue edges and vertices of Figure \ref{fig: twins}).
We resort to $T'$ to ensure that no two vertices or edges of the same triangle are identified. 
One can easily update the representation $E_T(\gamma)$ to $E_{T'}(\gamma)$ in time $O(|T|)$.
Because $T'$ is assumed to be oriented, for each edge $e_{i}\in E_{T'}(\gamma)$, we can get, in time $O(1)$ on our data structure, the triangle $t_i$ to the left of $e_\gamma$ according to the traversing direction of $\gamma$ as defined by $E_{T'}(\gamma)$.
With the information of $t_i$, we can query the other two edges of $t_i$ distinct from $e_{i}$, $f_{i,1},f_{i,2}$, and the indices of the intersections $k_{i,1},k_{i,2}=\pm 1$ of a directed arc in $t_i$ parallel to $e_{i}$.
We also denote by $v_i$ the unique vertex of $t_i$ incident to $e_i$ and $e_{i+1}$ (uniqueness is guaranteed by the barycentric subdivision).
We will iterate over the tuples $(f^{k_{i,1}}_{i,1},f^{k_{i,2}}_{i,2}, v_i, t_i)$, $1\leq i \leq \|E_T(\gamma)\|$, to construct an uncompressed intersection sequence, $w_\gamma$, of a standard curve $\gamma'$ (purple in Figure \ref{fig: twins}) isotopic to $\gamma$.

\begin{figure}[htbp]
    \centering

    \begin{subfigure}[b]{\textwidth}
        \centering
        \includegraphics[width=0.6\textwidth]{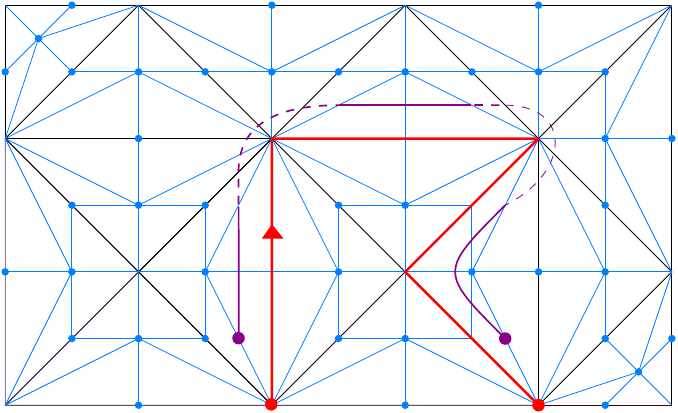}
        \caption{The triangulation $T$ (black) and the new edges and vertices introduced by its barycentric subdivision (blue). In red, we show a section of the edged curve $\gamma$ and the orientation induced by $E_T(\gamma)$. In purple, we highlight the standard curve $\gamma'$ isotopic to $\gamma$, where the solid intervals are those obtained by displacing the edges of $E_T(\gamma)$ and the dashed are those of form $(f')^j_{i,1},(f')^j_{i,2}$.}
        \label{fig: twins}
    \end{subfigure}

    \vspace{1em} % optional spacing between the images

    \begin{subfigure}[b]{\textwidth}
        \centering
        \includegraphics[width=0.6\textwidth]{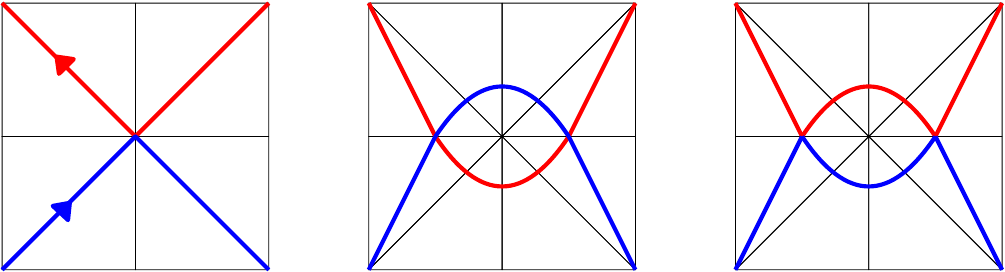}
        \caption{Left: The possible intersection of the components $c_1$ and $c_2$. Middle: the tracing induced by $w_1$ and $w_2$ before the removal of bigons. Right: the surgery to remove bigons between $c_1$ and $c_2$.}
        \label{fig: bigon}
    \end{subfigure}

    \caption{}
\end{figure}

Attach $f^{k_{i,1}}_{i,1}$ to $w_{\gamma}$. Call $(t')_i^1$ the triangle adjacent to $t_i$ by the edge $f_{i,2}$; note that $(t')_i^1$ is also incident to $v_i$.
Again, because of the barycentric subdivision, there is exactly one edge $(f')_{i,2}^1$ distinct from $(f')_{i,1}^1=f_{i,2}$ incident to $v_i$.
The pair $((f')^1_{i,1}, (f')^1_{i,2})$ defines an elementary arc in $(t')_i^1$ given by an intersection sequence $((f')^1_{i,1}, (f')^1_{i,2})$ (we here omit the intersection indices to simplify the notation), see the dashed purple lines in Figure \ref{fig: twins}. 
Similarly and recursively, the triangle $(t')^2_i$ that shares $(f')^1_{i,2}$ with $(t')^1_i$ has exactly one other edge incident to $v_i$, call it $(f')^2_{i,2}$. 
Because $f^{k_{i,1}}_{i,1}$ is incident to $v_i$, we can repeat the process until $(f')^n_{i+1,2}=f^{k_{i,1}}_{i,1}$ for some $n\leq 3|T'|$.
We then append $(f')^1_{i,1}, (f')^2_{i,1},\dots, (f')^n_{i,1}, f_{i+1,1}$ to $w_\gamma$. 
Since the curve $\gamma$ is assumed simple, no edge or vertex occurs twice in the list $E_T(\gamma)$, implying that this process takes time $O(|T'|)=O(|T|)$. 
We then remove edges in $T'-T$ from $w_\gamma$ and cyclically reduce it to ensure that it represents a normal curve.
Finally, $w_\gamma$ can be transformed into a trivial SLP, $I_T(\gamma)$, in time $O(|w_\gamma|)=O(|T|)$. 

Now, assume that $\#\gamma\geq 2$ and suppose we separately apply the above procedure to each component $c_1,\dots,c_{\#\gamma}$ of $\gamma$, obtaining uncompressed sequences $w_1,\dots,w_{\#\gamma}$.
Although these sequences individually represent curves that are isotopic to the components of $\gamma$, tracing them might create intersections between two components, say $c_1$ and $c_2$, if $E_T(c_1)$ and $E_T(c_2)$ shared a vertex, as shown in Figure \ref{fig: bigon}.
One can detect these cases in time $O(\#\gamma\times \|E_T(\gamma)\|)$ and solve the intersections by exchanging the parts of the sequence $w_1$ with the parts of the sequence $w_2$ that cobound bigons, ultimately performing the surgery in the bottom of Figure \ref{fig: bigon}.
We once more cyclically reduce each curve $w_1,\dots, w_{\#\gamma}$ (note that the cyclic reductions cannot form new bigons between components) and compute the trivial SLPs.
\end{proof}

\begin{theorem}\label{th: word to diagram}
    Suppose $\phi=\tau^{k_n}_{s_n}\circ \dots \circ \tau^{k_1}_{s_1}$, $s_i\neq s_{i+1}$ for $1\leq i< n$, is a Heegaard word in $\Mod(\Sigma_g)$ and let $T$ be a triangulation of $\Sigma_g$ with $\alpha$ and Lickorish generators edged. 
    There is an algorithm to compute intersection sequences $I_{T}(\phi(\alpha))$ in time $\Theta(g|T|+n|T|\lg k)$ in a unit-cost RAM model of bit-size $\lg k$ where $\lg k = \sum_{i=1}^n \lg k_i$. 
    Moreover, the Heegaard diagram $(T, E_{T}(\alpha),N_{T}(\phi(\alpha)))$ of complexity $O(|T|^3\lg^2 k)$ can be computed in time $O((|T|\lg k)^{10})$.
\end{theorem}
\begin{proof}
   For each $s\in \mathcal{L}$, denote its edge list by $E_T(s)=\{e_{s,1},\dots,e_{s,\|E_T(s)\|}\}$, where $e_{s,j}\in E$ for $1\leq j\leq \|E_T(s)\|$, and define $I_T(s,j)$ as (an SLP of) the intersection sequence of a normal curve isotopic to $s$ starting at a face adjacent to $e_{s,j}$. 
   We can use Lemma \ref{lm: get twins} to compute $I_T(s,j)$ for each $1\leq j\leq  \|E_T(s)\|$ in time $O(g|T|)$.
   We also compute $I_T(a)$ for each $a\in\alpha$; in total, this initial step takes time $O(g|T|)$. 
   
    Let $I_T(\beta_0)=I_T(\alpha)$.
    Assuming $k_i>0$, we recursively construct $I_T(\beta_i)$, for $1\leq i \leq n$, by replacing any occurrence of an edge $e_{s_i,j} \in E_T(s_i)$ in the SLPs $I_T(\beta_{i-1})=\{\langle x_1=\text{EXPR}_1,\dots, x_p=\text{EXPR}_p\rangle\}_{b\in\beta}$ with an SLP $\langle y_1,\dots,y_{q}\rangle$ of $I_T(s_i)^{k_i}\cdot s_{i,j}$, that is,
    \begin{equation*}
        I_T(\beta_i) = \{\langle  y_1,\dots,y_{q}, x'_{q+1}=\text{EXPR}'_1,\dots, x'_{q+p}=\text{EXPR}'_q \rangle\}_{b\in\beta} 
    \end{equation*}
    where $\text{EXPR}'_r$ is the word $\text{EXPR}_r$ with each occurrence of $x_t$ substituted by $x'_{p+t}$ and each occurrence of $s_{i,j}$ substituted by $y_q$. 
    The case of $k_i<0$ is treated similarly. 
    By Lemma \ref{lm: slp}, we note that this can be done in time $O(\|I_T(\beta_{i-1})\|+\lg k_i\|I_T(s_i)\|)$ in the assumed model of computation.
    In particular, $I_T(\beta_i)$ is an intersection sequence representation of $\beta_i=\tau^{k_i}_{s_i}\circ \dots \circ \tau^{k_1}_{s_1}(\alpha)$ and has complexity
    \begin{equation*}
    \begin{split}
        O(\|I_T(\beta_i)\|)&=O(\|I_T(\beta_{i-1})\|+\lg k_i\|I_T(s_i)\|)\\
        &= O\bigg(\|I_T(\alpha)\|+\sum_{i'=1}^{i}\lg k_{i'} \|I_T(s_{i'})\|\bigg)\\
        &\leq O\bigg(|T|+|T|\sum_{i'=1}^{i}\lg k_{i'} \bigg)=O\bigg(|T|\sum_{i'=1}^{i}\lg k_{i'} \bigg),
    \end{split}
    \end{equation*}
    where in the second line we recursively applied the equation of the first.
    The total time to compute $I_T(\beta)=I_T(\beta_n)$ is bounded by $n\times \|I_T(\beta)\|=O(n|T|\lg k)$, and, in the case $k_i=1$ for all $1\leq i \leq n$, this becomes $\Theta(n^2|T|)=\Theta(n|T|\lg k)$.

    For the second part, note that while $I_T(\beta)$ is, by construction, standard in $T$, it might not be normal.
    We can use the main algorithm of \cite{miyazaki1997improved} to deterministic compute, in time $O(\|I_T(\beta)\|^{10})$, cyclic reduced SLPs for $\beta$ with complexity $O(\|I_T(\beta)\|^2)$. 
    Alternatively, this can be done using a randomized algorithm in time $O(\|I_T(\beta)\|^{3})$ with some probability of error at most $\delta$ and in a unit-cost RAM model of bit-size $O(\|I_T(\beta)\|+\lceil1/\delta\rceil)$.
    Part (c) of Lemma \ref{lm: slp} finishes the proof.
\end{proof}

\section{Algorithms on Heegaard diagrams}\label{sec: algorithms}
We here describe some polynomial time algorithms and operations on a genus $g$ Heegaard diagram $(T, E_T(\alpha),N_T(\beta))$ with complexity $m$ and representing a closed 3-manifold $M$. 
For convenience, the algorithms and their complexity are summarized in Table \ref{tb: algorithms}.
For the reported times, we assume a unit-cost RAM model of bit-size $\log |\beta \cap E|=O(\|N_T(\beta)\|/|T|)$. 
\begin{theorem}[\texttt{CHECK DIAGRAM}]\label{th: check diagram}
    There is an algorithm to check, in time $O(m|T|)$, whether a normal multicurve $\beta$ of complexity $m$ defines a Heegaard diagram.
\end{theorem}
\begin{proof}
    Use a generalization of Theorem \ref{th: erickson}~\cite{erickson2002optimally}[Theorem 6.5] to trace the street complex $S(T,\beta)$, this naturally induces triangulations of the components in $T\backslash \beta$, which, following \cite{Erickson2013}, we call \textit{pieces}.
    We compute the Euler characteristic of each piece in time $O(|T|)$ \cite{erickson2002optimally}, ultimately giving their topological type. 
    The vectors $N_T(\beta)$ define a Heegaard splitting if and only if the only piece is a $2g$ punctured sphere.
\end{proof}

\begin{theorem}[\texttt{GET EFFICIENT POSITION}]\label{th: get efficient position}
    There is an algorithm to isotope $\alpha$ and $\beta$ to efficient position in $\Sigma_g$ in time $O(\poly(m,|T|))$. 
\end{theorem}
\begin{proof}
    Lackenby describes an algorithm for this problem in Theorem 6.3 of~\cite{lackenby2024some}, where we implicitly use the content of his Corollary 6.5, namely that the number of bigons between the $\alpha$ and $\beta$ curves is upper bounded by $154|T|$.
\end{proof}

\begin{theorem}[\texttt{CHECK ISOTOPY}]\label{th: check isotopy}
    There is an algorithm to check, in time $O(\poly(m,|T|))$, whether two Heegaard diagrams, $\beta$ and $\beta'$, of complexity at most $m$, are isotopic.
\end{theorem}
\begin{proof}
    Once more, this follows from Lackenby's paper \cite{lackenby2024some}, this time, from Theorem 1.2.
\end{proof}

\begin{localsize}{10}
\begin{table}
\renewcommand{\arraystretch}{1.5}
\begin{tabular}{ p{3.5cm} p{5cm} c c}
\hline
Algorithm name & Description & Complexity & Reference \\  
\hline
 \texttt{CHECK DIAGRAM}& Decide if normal coordinates define a diagram& $O(m|T|)$ & Theorem \ref{th: check diagram}\\
 \texttt{GET EFFICIENT POSITION}& Put $\alpha$ and $\beta$ curves in efficient position & $O(\poly(m,|T|))$ & Theorem \ref{th: get efficient position}\\
 \texttt{CHECK ISOTOPY}& Decide if two diagrams of complexity $m$, on the same marked triangulation, are isotopic& $O(\poly(m,|T|))$& Theorem \ref{th: check isotopy}\\
  \texttt{DO A PACHNER MOVE}& Do a Pachner move in the underlying triangulation& $O(1)$ & Theorem \ref{th: do a pachner move}\\
 {\texttt{DO A DISK SLIDE}}& Slide a disk $b_i$ over a disk $b_j$& $O(m|T|)$ & Theorem \ref{th: do a disk slide}\\
{\texttt{DO STABILIZATION}}& Compute a stabilization of the diagram& $O(1)$ & Theorem \ref{th: do stabilization}\\
{\texttt{DO DESTABILIZATION}}& Detect if there is a trivial stabilization pair and compute destabilization& $O(g|T|)$& Theorem \ref{th: do destabilization}\\
{\texttt{DETECT REDUCTION}}& Find either a common component of $\alpha$ and $\beta$ or a separating curve disjoint from both& $O(\poly(g,m,|T|))$ & Theorem \ref{th: detect reduction}\\
{\texttt{DO A DEHN TWIST}}& Compute the action of a Dehn twist about a normal curve $s$ with complexity $m'$ on $\beta$ & $O(((m+m')|T|)^{10})$& Theorem \ref{th: do a dehn twist}\\
\texttt{GET }$\pi_1$ & Compute a presentation of the fundamental group of $M$ & $O(m|T|)$& Theorem \ref{th: get pi1}\\
\texttt{GET }$H_1$ & Compute a presentation of the first homology group of $M$& $\widetilde{O}(m^2g^3+m|T|)$& Theorem \ref{th: get h1}\\
 \hline
\end{tabular}
\caption{}\label{tb: algorithms}
\end{table}
\end{localsize}

\begin{figure}
    \centering
    \includegraphics[width=0.45\textwidth]{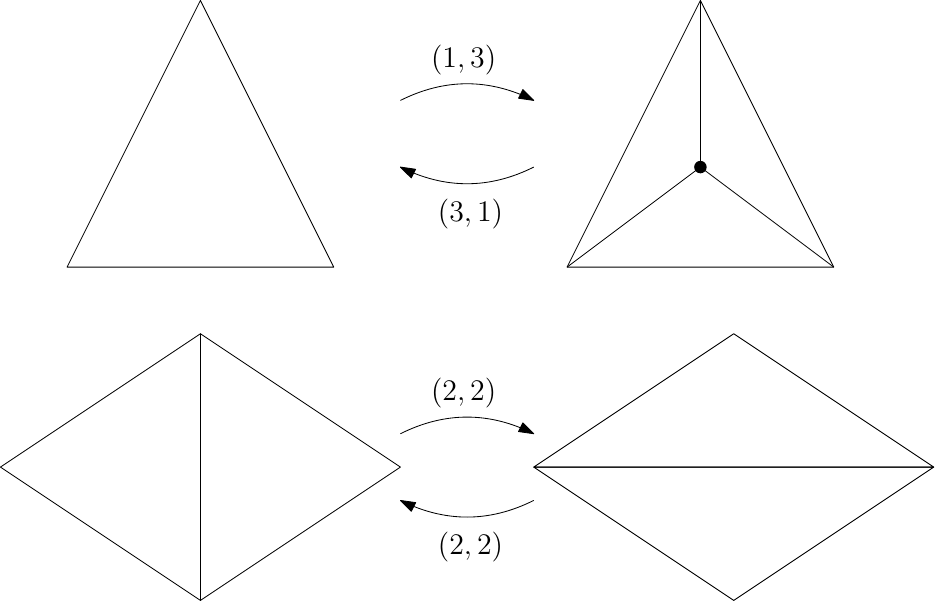}
     \caption{}
  \label{fig: pachner}
\end{figure}

\begin{theorem}[\texttt{DO A PACHNER MOVE}]\label{th: do a pachner move}
    A diagram $(T',E_{T'}(\alpha),N_{T'}(\beta))$, where $T'$ is the result of a Pachner move applied to any set of triangles in $T$, can be computed in time $O(1)$.
\end{theorem}
\begin{proof}
First, let $\gamma$ be a normal (multi)curve with respect to a triangulation $T$. 
Having fixed a triangle in $T$, we call its \emph{corner coordinates} the number of elementary arcs of each type.
In particular, note that for the triangle $t$ with sides $e$, $f$, and $g$, the number of elementary arcs between sides $e$ and $f$ is $1/2\times[e(\gamma)+f(\gamma)-g(\gamma))]$ and symmetrically for the other cases.
Given $e(\gamma), f(\gamma)$, and $g(\gamma)$, one can use these formulas to compute the corner coordinates of $t$ in time $O(1)$.

Now refer to the two moves of Figure \ref{fig: pachner}. 
These modifications of the underlying triangulation are known as \emph{Pachner moves} and have received special attention in low-dimensional topology because any two triangulations are related by finitely many applications of them~\cite{pachner1991pl}. 
Here, we assume that no such move deletes an edge in $E_T(\alpha)$; otherwise, we would not have a marked triangulation as an output.
Move $(3,1)\to(1,3)$ corresponds to the simple deletion of the relevant components of $N_T(\beta)$, move $(1,3)\to (3,1)$ involves adding new components whose values are corner coordinates. 
Both are implemented in time $O(1)$ in our data structure for $T$. 
Similarly, moves $(2,2)\to(2,2)$ can be implemented in time $O(1)$ directly by computing the corner coordinates and updating the triangulation. 
\end{proof}

\begin{theorem}[\texttt{DO A DISK SLIDE}]\label{th: do a disk slide}
    A diagram $(T,E_T(\alpha),N_T(\beta'))$, where $\beta'$ is the result of applying a disk slide on any two components of $\beta$, can be computed in time $O(m|T|)$.
\end{theorem}

\begin{proof}
    Compute the street complex $S(T,\beta)$; recall that it has complexity of $O(|T|)$. 
    Using breadth-first in the dual graph of $S(T,\beta)$, find, in time $O(|T|^2)$, the shortest arc $p$ that connects faces adjacent to $b_1$ and $b_2$ and that does not cross any component of $\beta$, which can be avoided by adding infinite weights to the edges of $S(T,\beta)$ not in $T$. 
    Note that such an arc must exist because no component of $\beta$ can be separating in the surface. 
    Moreover, because $p$ is shortest and, therefore, intersects at most $O(|T|)$ streets in $S(T,\beta)$, we can compute, using Part (b) of Theorem \ref{th: erickson}, an SLP of complexity $O(m|T|+|T|)=O(m|T|)$ of an intersection sequence of $p$ with respect to $T$, call it $I_T(p)$.
    
    Suppose that $p$ is not trivial, that is, it is not just a single point.
    By Theorem 6.2 of \cite{Erickson2013}, find, in time $O(m|T|)$, SLPs of the intersection sequences of $b_1$ and $b_2$ with respect to $T$ starting at edges adjacent to the endpoint faces of $p$, call them $I_T(b_1)$ and $I_T(b_2)$.
    Finally, let $I_T(b)=I_T(b_2)^{-1}\circ I_T(p)^{-1}\circ I_T(b_1)\circ I_T(p)$.  
    Note that $I_T(b)$ already is an intersection sequence of a normal curve in $T$.
    For, since $p$ is a shortest path, it is normal in $S(T,\beta)$, which implies that it is also normal in $T$.
    Moreover, $p$ has junctions as endpoints; if it had streets as endpoints, there would be a shorter arc connecting $b_1$ to $b_2$ starting at one of the ends of the street.
    Therefore, any other reduction in $I_T(b)$ would necessarily arise by $p$ entering a junction of $S(T, \beta)$ through a non-redundant port $e$ that is also intersected by either $b_1$ or $b_2$.
    But then, the other triangle in $T$ adjacent to $e$ is adjacent to either $b_1$ or $b_2$, implying that $p$ is not shortest. 
    By repeated applications of Part (c) of Lemma \ref{lm: slp}, one can compute, in time $O(|T|\times m)$, the number of (unsigned) occurrences of each edge of $E$ in the SLP of $I_T(b)$, resulting in $N_T(b)$. 

    If $p$ is trivial, then $b_1$ and $b_2$ bound the same face in $S(T,\beta)$. 
    Nonetheless, because the curves are not isotopic, there exists at least one face bounded by $b_1$ and not bounded by $b_2$. 
    We call such a face \emph{good}. 
    Choose, in time $O(|T|)$, a non-good face $t_1$ that is adjacent to a good face through an edge $e_1$, see Figure \ref{fig: stabilization}.
    Note that $e_1$ is the endpoint of a street $s$ of sides in $b_1$ and $b_2$. 
    Look for the other end of this street, $e_2$---again, because $b_1$ and $b_2$ are not isotopic, the street has another end. 
    Let $t_2$ be the face (in fact, a junction) not equal to $s$, but also adjacent to $e_2$. 
    Compute, again through~\cite{Erickson2013}[Theorem 6.2], the sequence of faces adjacent to $b_1$ starting at $t_1$ and going all the way up to $t_2$ without crossing $s$, call this sequence $I_T(\widetilde{b_1})$.
    Compute the related sequence for $b_2$, this time stating in $t_2$, call it $I_T(\widetilde{b_2})$.
    Let $I_T(b) = I_T(\widetilde{b_2})\circ I_T(\widetilde{b_1})$, note that it is also normal in $T$ as $b_1$ and $b_2$ are normal and it enters $t_1$ and $t_2$ through distinct edges.
    Part (c) of Lemma \ref{lm: slp} finishes this case.
\end{proof}

\begin{figure}
    \centering
    \includegraphics[width=0.8\textwidth]{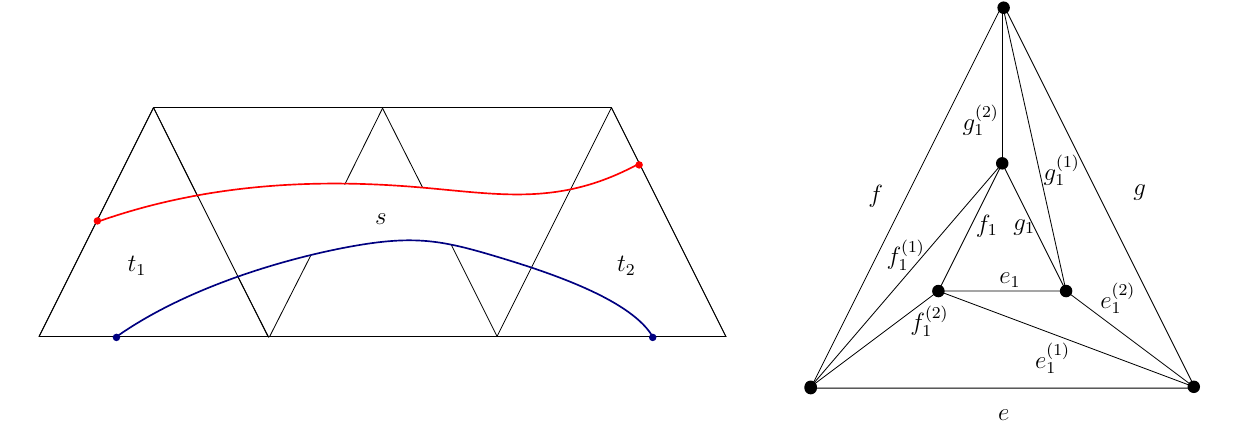}
     \caption{Left: diagram for the proof of Theorem \ref{th: do a disk slide}. The red and blue curves represent $b_1$ and $b_2$, the black lines are the other edges of $S(T,\beta)$. The faces $s$, $t_1$, and $t_2$ are also indicated. Right: diagram for the proof of Theorem \ref{th: do stabilization}. Figure inspired by \cite{ennes2025hardnesscomputationquantuminvariants}.}
  \label{fig: stabilization}
\end{figure}

\begin{theorem}[\texttt{DO STABILIZATION}]\label{th: do stabilization}
    A diagram $(T',E_{T'}(\alpha\cup\{a\}),N_{T'}(\beta\cup\{b\}))$ of a stabilization of $(T,E_T(\alpha),N_T(\beta))$ can be computed in time $O(1)$.
\end{theorem}

\begin{proof}
    A similar algorithm for (uncompressed) intersection sequences as input is described in \cite{ennes2025hardnesscomputationquantuminvariants}; here we extend it to our data structure.
    Choose any face $t$ of $T$ and modify it by adding a triangle $t_1$ of sides $e_1,f_1,g_1$ to its center, together with 6 other edges $e^{(1)}_1,e^{(2)}_1,f^{(1)}_1,f^{(2)}_1,g^{(1)}_1,g^{(2)}_1$ connecting the vertices of $t_1$ to the vertices of $t$, see Figure \ref{fig: stabilization}.
    Define the triangulation $T_1$ as $T$ with $t_1$ deleted and update the normal coordinates of $\beta$ to be a vector in $\mathbb{N}^{|T|+9}$, where the new components are set to
\begin{equation*}
\begin{split}
    &e_1(\beta), f_1(\beta), g_1(\beta) \xleftarrow{} 0 \;\;\;e^{(1)}_1(\beta),e^{(2)}_1(\beta)\xleftarrow{}\frac{1}{2}[f(\beta)+g(\beta)-e(\beta)]\\ &f^{(1)}_1(\beta),f^{(2)}_1(\beta)\xleftarrow{}\frac{1}{2}[g(\beta)+e(\beta)-f(\beta)]\;\;\;g^{(1)}_1(\beta),g^{(2)}_1(\beta)\xleftarrow{}\frac{1}{2}[e(\beta)+f(\beta)-g(\beta)].
\end{split}
\end{equation*}
Note that $T\backslash t_1$ is homeomorphic to $\Sigma_g$ with one puncture. 

Now let $T_2$  be a fixed triangulation of the punctured torus with a triangular boundary component $t_2$, endowed with an edged curve $a$ and a normal curve $b$, both disjoint from $t_2$. 
Glue $T_1$ to $T_2$ by identifying $t_1$ with $t_2$. 
This forms a triangulation of $\Sigma_{g+1}$, with $\alpha\cup \{a\}$ and $\beta\cup \{b\}$ naturally giving a Heegaard diagram of a stabilization of $(\Sigma_g,\alpha,\beta)$.
\end{proof}

\begin{theorem}[\texttt{DO DESTABILIZATION}]\label{th: do destabilization}
    One can find a trivial stabilization pair of the diagram $(T,E_T(\alpha),N_T(\beta))$ in time $O(g|T|)$ and, if found, a destabilization of the diagram can be computed in constant time.
\end{theorem}
\begin{proof}
    For each $a\in \alpha$, check the normal vectors of $N_T(\beta)$, in time $O(g\|E_T(\alpha)\|)$, for an $a$ that intersects only a single component of $\beta$ and only once. 
    Suppose that there is such an $a$ and that $a,b$ is the detected trivial stabilization pair. 
    In time $O(\|E_T(a)\|)=O(|T|)$, cut $T$ at $a$ and attach disks to the holes, therefore giving a triangulation of the surface of genus $g-1$.
\end{proof}

\begin{theorem}[\texttt{DETECT REDUCTION}]\label{th: detect reduction}
    There is an algorithm that runs in time $O(\poly(m,|T|))$ and detects whether a diagram is reducible. 
    Moreover, if the diagram is reducible, it returns either an isotopic pair of $a\in \alpha$ and $b\in \beta$, or a separating curve $s$.
\end{theorem}

\begin{proof}
 We divide the algorithm into two parts, corresponding to the two hypotheses of reducibility of diagrams, namely the existence of an isotopic pair of meridians $a\in \alpha$ and $b\in \beta$, or of an essential separating curve that splits $\Sigma_g$ into subdiagrams. 
 For the first case, apply \cite{lackenby2024some}[Theorem 6.3] to find, in polynomial-time on the input, an isotopic diagram $(\alpha,\beta')$ in efficient position.
 This procedure increases the complexity of $\beta$ to $m'= m(1+2|T|)$
 We note that although this subroutine is polynomial-time on the input, the author of~\cite{lackenby2024some} does not provide an estimation of the polynomial degree, and any attempt on our part to give one would be out of the scope of this paper. 
 Because we made the diagram efficient, there can only be a component $a\in\alpha$ isotopic to some $b'\in \beta'$ if $e(b')=0$ for all $e\in E_T(\alpha)$.
 If that is the case, we trace $S(T,b')$. 
 Since $b'\cap a =\emptyset$, the curve $a$ is still fully in the edges of the street complex. 
 Nonetheless, $a$ is isotopic to $b'$ if and only if they cobound a cylinder \cite{farb2011primer}, a case that can be determined in time $O(|S(T,b')|)$. 
 Checking this for every pair of curves $a$ and $b'$ that do not intersect---there can be at most $g^2$ of these pairs---gives a test for the first hypothesis.

 \begin{figure}
    \centering
    \includegraphics[width=0.8\textwidth]{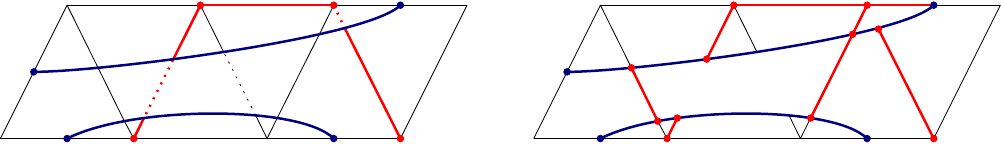}
     \caption{Diagram illustrating the proof of Theorem \ref{th: detect reduction}. Left: part of the overlay graph defined by the $\beta$ curves (blue) with respect to a triangulation (black) and the corresponding embeddings of a segment of the $\alpha$ curves (red). Redundant ports are shown in dashed. A cellular embedding of the surface $\Sigma''$ can be obtained by puncturing the diagram along both blue and red curves. Right: another cellular embedding for $\Sigma''$ in which we enlarge the holes punctured when cutting $\alpha$ by moving redundant ports in $E_T(\alpha)$ to both ends of the streets.}
  \label{fig: cutting}
\end{figure}

 If the above test returns no clear reduction pair, we look for a separating curve, $s$, disjoint from $\alpha$ and $\beta$, which exists if and only if $s$ is non-contractable in $\Sigma'' = \Sigma_g\backslash(\alpha\cup \beta)$. Because $\Sigma'=\Sigma_g\backslash \beta$ is a $2g$ puncture sphere, $\Sigma''$ is a union of (potentially one) punctured spheres, and it embeds a non-contractable curve if and only if one of its components is not a disk.
 This means that we only need to compute a cellular embedding of $\Sigma''$ and check its topological type. 
 We start computing a triangulation for $\Sigma'$ as in the proof of Theorem \ref{th: do a disk slide} by tracing $S(T,\beta)$ and disconnecting from the gluing rules the edges of $E_{S(T,\beta)}(\beta)$.
 We then want to further cut $\Sigma'$ along $\alpha$.
 But, because $\alpha$ is edged with respect to $T$, each segment of $E_T(\alpha)$ will be a port of the overlay graph of $\beta$ (with respect to $T$).
 Each such a port will either be non-redundant---in which case it will still be an edge of $S(T,\beta)$)---or it will be redundant---in which case it will be in the crossing sequence of some non-degenerate street and will be deleted to form $S(T,\beta)$ (left side of Figure \ref{fig: cutting}). 
 Therefore, for each street in $E_{S(T,\beta)}(\beta)$ (degenerate or not), we can use Part (c) of Lemma \ref{lm: slp} and Part (b) of Theorem \ref{th: erickson} to determine, in time $O(\|E_T(\alpha)\|\times m|T|)$, if there is an occurrence of an edge in $E_T(\alpha)$ in a street's crossing sequence and delete from the complex $S(T,\beta)$ (the interior of) any such a street. Because the streets are disks in $\Sigma'$, bounded by one or two segments of $\beta$, we note that, topologically, deleting the full interior of the streets that are transverse to some edge in $E_T(\alpha)$ is equivalent to puncturing a hole corresponding to a redundant port of the overlay graph by simply ``enlarging" the hole (right side of Figure \ref{fig: cutting}). This means that we have a cellular embedding of $\Sigma''$ which can be queried, in polynomial-time, for its topological type. Assuming that $\Sigma''$ has a component that is not a disk: we can use the techniques of \cite{erickson2002optimally} (refer to Section 2.3 of \cite{chambers2006splitting} for details) to find, in time $O(|T'|)$ a non-contractable curve $s$ in $\Sigma''$ expressed as an SLP of complexity $O(m|T|+|T|)=O(m|T|)$.
 We finally compute $S(T,s)$, in time $O(m|T|^2)$, to check whether a diagram is defined on each component of $\Sigma_g\backslash s$. 
 For such, we need only to find 1. the number of $\alpha$ and $\beta$ components on each side of $s$; and, 2. whether $s$ is not isotopic to an $\alpha$ or $\beta$ curve by looking for bounded cylinders, as described in the first case. 
\end{proof}

\begin{theorem}[\texttt{DO A DEHN TWIST}]\label{th: do a dehn twist}
    Given a normal curve $s$ in $T$ of complexity $\|N_T(s)\|=m'$, one can find, in time $O(((m+m')|T|)^{10})$, the diagram $(T,E_T(\alpha),N_T(\tau_s(\beta))$.
\end{theorem}
\begin{proof}
    This algorithm is described by \cite{schaefer2008computing, Stefankovic}. 
    If one requires only the intersection sequences $I_T(\tau_s(\beta))$, the algorithm can be executed in linear time on the inputs, but the output might just be standard with respect to the triangulation. 
    Therefore, as in the proof of Theorem \ref{th: word to diagram}, we are forced to use the subroutines described by the same authors to cyclic reduce---either in time $O(((m+m')|T|)^{10})$ with a deterministic algorithm or in time $O(((m+m')|T|)^{3})$ with a randomized one---the multicurve given by $I_T(\tau_s(\beta))$.
\end{proof}

\begin{theorem}[\texttt{GET} \texorpdfstring{$\pi_1(M)$}{pi1(M)}]\label{th: get pi1}
    A presentation of $\pi_1(M)$ can be computed in time $O(m|T|)$.
\end{theorem}
\begin{proof}
    Given $I_T(\beta)=\{I_T(b)\}_{b\in\beta}$ delete, in linear time on the complexity of the SLPs, all references to the edges $E-E_T(\alpha)$ and substitute each edge of $E_T(\alpha)$ by the symbol that represents its associated component $a\in \alpha$, here seen as a generator $\pi_1(M)$.
    This gives $g$ distinct SLPs $r_i = [x_1^i=\text{EXPR}_1^i,\dots, x_{\ell_i}^i=\text{EXPR}_{\ell_i}^i]$ whose complexities add up to $m$ and describe the relations of the presentation of $\pi_1(M)$.
    A presentation of $\pi_1(M)$ can then be defined with generators $a_1,\dots, a_g, x_1^1,\dots, x_{\ell_i-1}^1,\dots, x_{\ell_{g}-1}^g$ and relators $(x_j^i)^{-1}\cdot \text{EXPR}_j^i$ (this corresponds to repeated applications of a Tietze move~\cite{johnson1997presentations} to the presentation of $\pi_1(M)$ computed using the method described at the end of Section~\ref{sec: background topology}).
    The output is a balanced presentation of $\pi_1(M)$ with $\|I_T(\beta)\|$-many generators and relators. 
    The total time is dictated by Theorem \ref{th: erickson}.  
\end{proof}

\begin{remark}\label{rm: co-NP}
    Because the outputted presentation of $\pi_1(M)$ in Theorem \ref{th: get pi1} has linear size on the input, most complexity results about \emph{uncompressed presentations} of 3-manifolds (such as triangulations) that appeal to the fundamental groups still hold in our \emph{compressed} data structure. 
    This is the case, for example, of Zentner's proof that sphere recognition lies in co-$\texttt{NP}$~\cite{Zentner_2018} (assuming the generalized Riemann hypothesis).
    Similarly, the usual bounds on the volume \cite{cooper1999volume} and diameter \cite{white2001diameter} of hyperbolic 3-manifolds can be naturally translated to our notion of complexity.
\end{remark}

\begin{theorem}[\texttt{GET }\texorpdfstring{$H_1(M)$}{H1(M)}]\label{th: get h1}
    The Smith normal form of the presentation matrix of $H_1(M)$ can be found in time $\widetilde{O}(m^2g^3+m|T|)$, where $\widetilde{O}$ indicates some hidden polylogarithmic factors.
\end{theorem}
\begin{proof}
    We first compute the presentation matrix $K$ of the homology group whose coefficients are algebraic intersection numbers between the diagram's curves. 
    Because we assume that each $a_j$ is a list of edges $E_T(a_j)$, given an SLP $I_T(b_i)$, 
\begin{equation*}
    \hat{i}(b_i,a_j)=\sum_{e\in E_T(a_j)} e(I_T(b_i))-e^{-1}(I_T(b_i)),
\end{equation*}
    which, by Lemma \ref{lm: slp}, can be computed in time $O(\|I_T(b_i)\|\times \|E_T(a_j)\|)$.
    Therefore, referring to Theorem \ref{th: erickson} to compute $I_T(\beta)$, the total computational time of this procedure is $O\bigg(m|T|+\sum_{i=1}^g\sum_{j=1}^g\|I_T(b_i)\| \|E_T(a_j)\|\bigg)=O(m|T|+m|T|)$. 
    Finally, we reduce the matrix $K$ to its Smith normal form in time $\widetilde{O}(g^3m^2)$ by the Kannan-Bachem algorithm~\cite{sergeraert2024kannanbachemalgorithm}.
\end{proof}

\section{Examples and experiments}\label{sec: examples and experiments}
Given a Heegaard word presentation $(\Sigma_g,\phi)$ of a closed 3-manifold $M$, there exist algorithms to compute, in linear time on the length of $\phi$, a triangulation of $M$ with as many tetrahedra~\cite{alagic2011quantum, alagic2014quantum, brinkmann2001computing,koenig2010quantum}.
This implies, for example, that the first homology group of $M$, $H_1(M)$, can be computed in polynomial-time on $|\phi|$.
These algorithms should be contrasted with Theorem \ref{th: word to diagram}, which says that SLPs of the intersection sequence $I_T(\beta= \phi(\alpha))$ can be computed in linear time on the complexity of $\phi$ \emph{measured in power-notation form}.
Because the fundamental group of $M$ depends only on $I_T(\beta)$, algebraic information on splittings with significantly complicated gluing maps can be retrieved by Theorems \ref{th: get h1} and \ref{th: get pi1}.

Here, we explore the gains from the compressed data structure by comparing an early implementation of our methods with \texttt{SnapPy}'s \cite{SnapPy} \texttt{Twister} \cite{Twister} module, commonly used to compute triangulations from Heegaard words.
Since experiments are limited to Heegard words as input, there is no need to use proper normal coordinates or even triangulations of $\Sigma_g$; instead, we may compute the SLPs of the induced diagram with respect to general cellular embeddings, therefore avoiding the high cost of cyclic reductions. 
This guarantees that, for the problems analyzed, our theoretical complexity is at least equal to \texttt{Twister}'s. 
Even without these simplifying assumptions, we are still exponentially faster in the worst cases, as we show in the following section about lens spaces.
Similar computational gains for genus 2 splittings are shown in Section \ref{sec: genus 2}, where some statistical properties of the homology groups of random splittings~\cite{rivin2014statistics} are explored.

\subsection{Genus 1}\label{sec: genus 1}
\begin{figure*}[t]
    \subfloat[]{%
        \includegraphics[width=.48\linewidth]{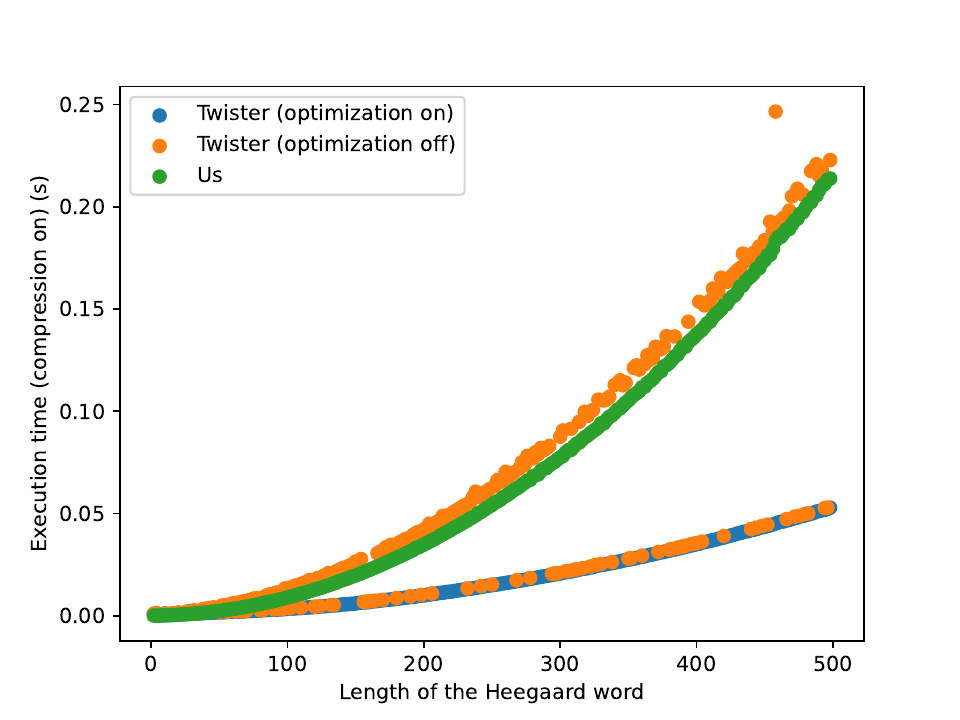}%
        \label{fig: robustness g 1}%
    }\hfill
    \subfloat[]{%
        \includegraphics[width=.48\linewidth]{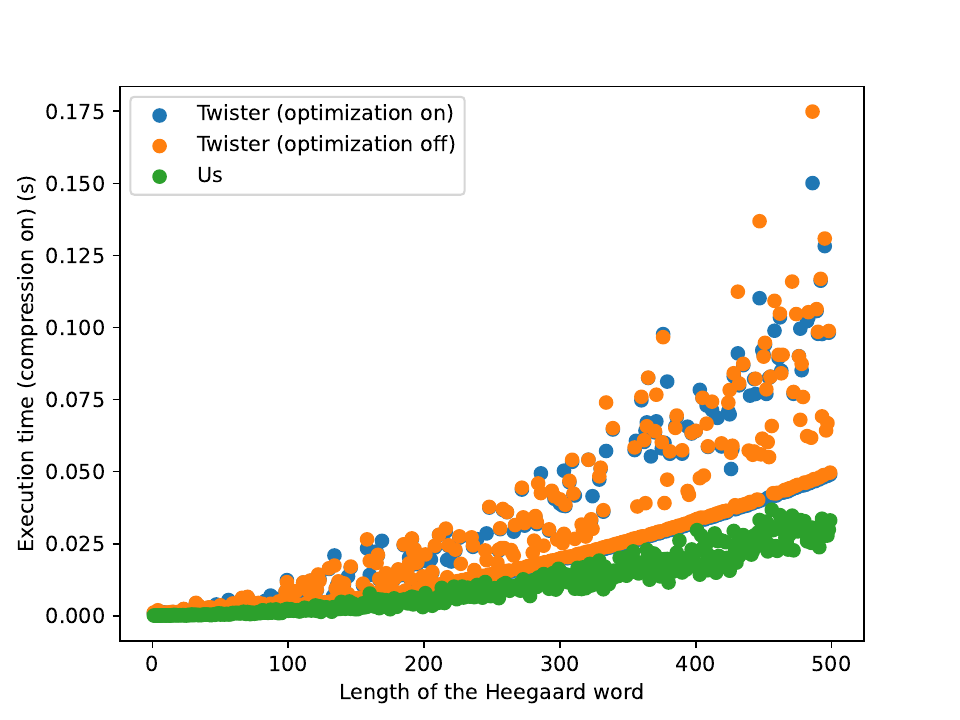}%
        \label{fig: time fibonacii g 1}%
    }\\
    \subfloat[]{%
        \includegraphics[width=.48\linewidth]{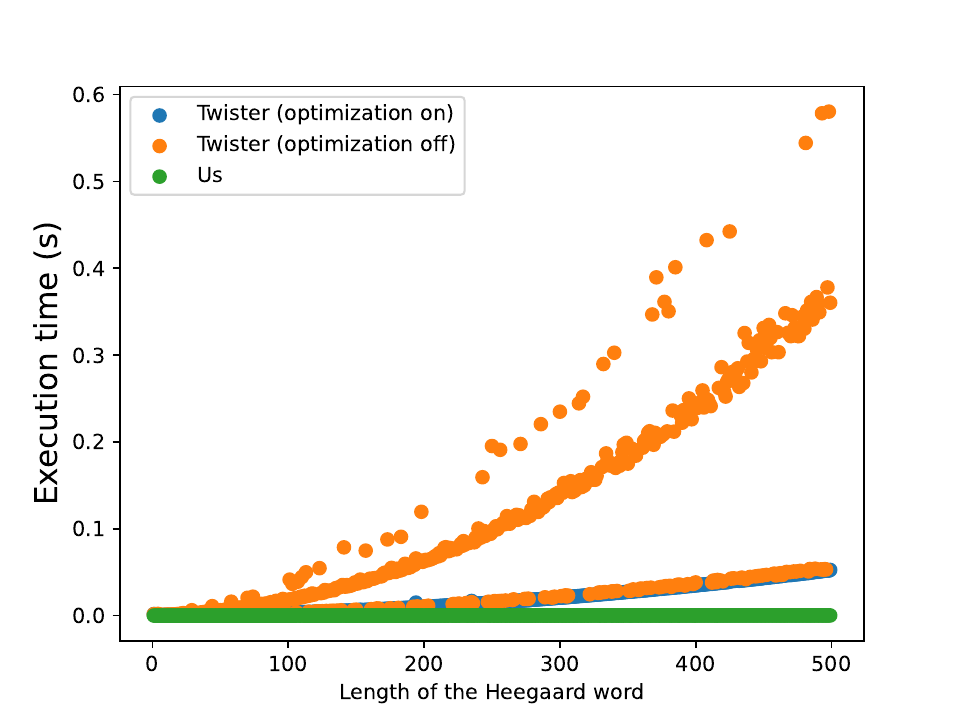}%
        \label{fig: time random g 1}%
    }\hfill
    \subfloat[]{%
        \includegraphics[width=.48\linewidth]{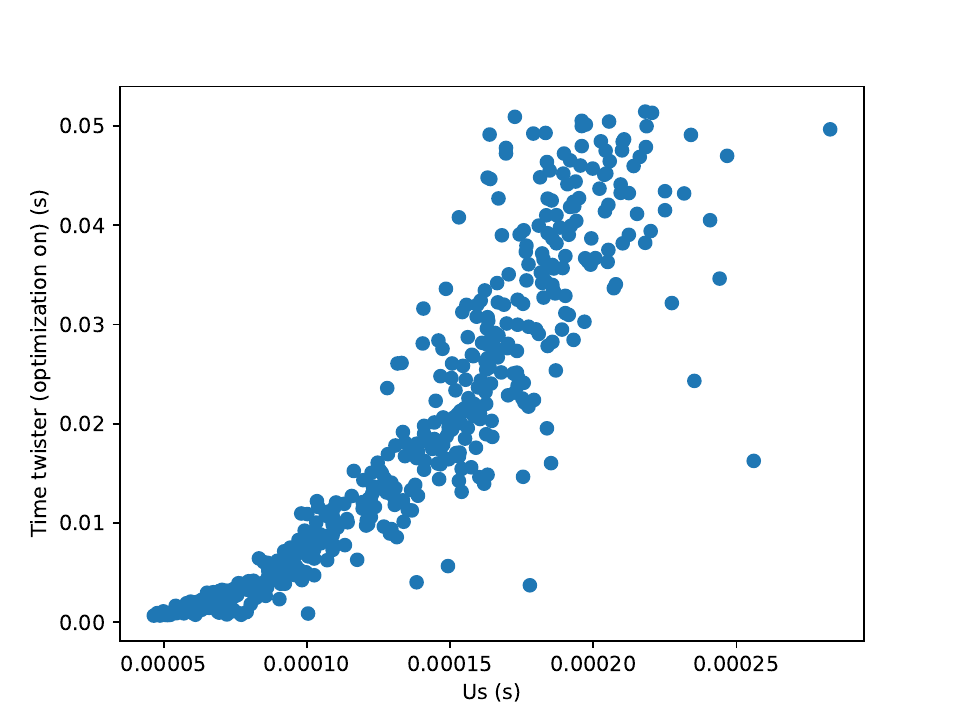}%
        \label{fig: time exponential g 1}%
    }
    \caption{Experiments with computing the homology group $H_1$ of Heegaard splittings over the torus. Averages were computed for 30 repetitions.}
    \label{fig: genus 1}
\end{figure*}

Heegaard splittings over the genus 1 surface---i.e., the torus---form the well-known family of \emph{lens spaces}. 
Every lens space is uniquely given by a pair of integers $(p,q)$, which, for a Heegaard splitting $(\Sigma_g,a, b)$, represents the algebraic intersection numbers of $b$ with the usual longitude and meridian of the torus, respectively (there is some ambiguity in the values of $p,q$, see~\cite{saveliev2011lectures}).  
A Dehn twist about the meridian $m$ of the torus maps $(p,q)$ to $(p + q,q)$, while about the longitude $\ell$ maps to $(p,q-p)$.
This gives a representation of $\Mod(\Sigma_1)$ in $\Z^2$ from which the homology of the Heegaard splitting $(\Sigma_g,\phi)$ can be exactly computed by $|\phi|$ matrix multiplications and one query on the final value of $q$.
We used this approach to benchmark the accuracy of our method, which returned the expected $H_1$ for input words $\phi$ of length up to 1000, the maximum size we tested.

Figure \ref{fig: genus 1} compares the average computational time needed to find $H_1(\Sigma_1,\phi)$ taken by our method and by $\texttt{Twister}$ for four different families of $\phi$ indexed by an integer $n$: (a) words of form $(\tau_m\tau_\ell^{-1})^n$; (b) uniformly sampled words of length $n$; (c) words of form $\tau_\ell^n$; and (d) words of form $\tau_\ell^n$, but this time, given in power-notation. 
Averages were taken over 30 repetitions.
In all cases, we compared our times with \texttt{Twister}'s, with and without optimization (where, by optimization, the authors mean the greedy folding of tetrahedra; refer to their documentation for details\footnote{Available at \url{https://snappy.computop.org}.}).
Although optimization drastically reduces computational time, it seems to increase numerical error.
In every instance, except for (d), we first transformed the words into a power-notation form before using them to compute $I_T(b)$; in total, this can be done in linear time on $|\phi|$.

\begin{remark}\label{rm: numerical error}
    A direct application of the representation of $\Mod(\Sigma_1)$ implies that, for the family of case (a), $H_1^n=\Z/f_{2n}$ where $f_n$ is the $n$-th Fibonacci number (starting at 0).
    We used this relation to benchmark the accuracy of the methods for different values of $n$. 
    While our technique always agrees with the expected value for up to $n=500$, \texttt{Twister} (with optimization both on and off) starts to deviate for values of $n\geq 25$.
    Interestingly, we did not observe divergences for higher genii splittings. 
\end{remark}

Case (a) is where our advantages compared to \texttt{Twister} are the worst.
This is because the words $(\tau_m\tau_\ell^{-1})^n$ do not admit a non-trivial power-notation form, which is mainly responsible for our gains in the other cases.
Nevertheless, it is worth highlighting that, while \texttt{Twister} is partially implemented in $\texttt{C++}$, our code is fully \texttt{python}.
Case (b) represents the average case of computing the homology of lens spaces given a gluing map $\phi$. 
We notice that we are now consistently faster for all values of $n$ tested.
This happens exactly as we explore redundancies by first putting the word in power-notation form, with significant gains every time the same Dehn twist occurs in sequence (since $\Mod(\Sigma_g)$ has two generators, this will be the case quite often).
The gains due to repetitions are further highlighted in (c): while \texttt{Twister} is bound to take $O(n^c)$, for some $c>1$, to compute the homologies, our approach takes time $O(n+\lg n)$, where the linear dependence comes solely from compressing the inputs to power-notation. 
In fact, if we assume the words already in power-notation, we are exponentially faster than \texttt{Twister}, as shown in case (d).

We can explore our structure's usage of power-notation to compute invariants of families of 3-manifolds that would necessarily require exponentially many tetrahedra to be represented.
The minimum number of tetrahedra in a \emph{layered triangulation} of a lens space $L(p,q)$ is $\sum_i(a_i)-1$, where $(a_0,\dots,a_k)$ is the continued fraction expansion of $q/p$ \cite{jaco2006layered}. 
Moreover, \cite{jaco2009minimal, jaco2011coverings} discovered families of lens spaces---including the spaces defined by $p=2n, q=1$ and $p=4n, q=2n-1, n\geq2$---for which layered triangulations are minimal among \emph{all} triangulations.
So, for example, while the family $L(2^n,1)$ necessarily requires exponentially many tetrahedra to be represented, it can be encoded in linear space in our data structure.
It seems natural to expect that there are many other families of manifolds, with diagrams exponentially smaller than the minimum number of tetrahedra, for which our techniques can provide real computational gains compared to any existing method.

\subsection{Genus 2}\label{sec: genus 2}
For the case of Heegaard splittings over the double torus, there is no easy version of the representation trick to compute the homology group, which means that we lack effective accuracy benchmarks.
However, this time, we can consider a \emph{random Heegaard splitting model} \cite{aschenbrenner20153}.
Formally, suppose we uniformly draw words from the radius-$n$ ball $B(n)=\{\phi \in \Mod(\Sigma_g)||\phi| = n\}$ for a fixed genus $g\geq 2$. 
We say that the class of 3-manifolds is \textit{asymptotically} $\mathcal{P}$, where $\mathcal{P}$ is a property of 3-manifolds (e.g., hyperbolic, Haken, etc.), if, for all $g\geq 2$, the probability that a manifold with splitting $\phi \sim \text{Uniform}(B(n))$ has the property $\mathcal{P}$ tends to 1 as $n\to\infty$.
While it has been shown, for example, that manifolds are asymptotically hyperbolic \cite{lubotzky2016random, https://doi.org/10.1112/jtopol/jtq031}, rational homology spheres with torsion \cite{dunfield2006finite}, and have particular features associated with some of their quantum invariants \cite{dunfield2011quantum}, other asymptotic properties, such as non-Hakeness and bounds on the hyperbolic volume \cite{rivin2014statistics}, remain only conjectured. 

\begin{figure}
\begin{subfigure}{.475\textwidth}
  \centering
  \includegraphics[width=.85\linewidth]{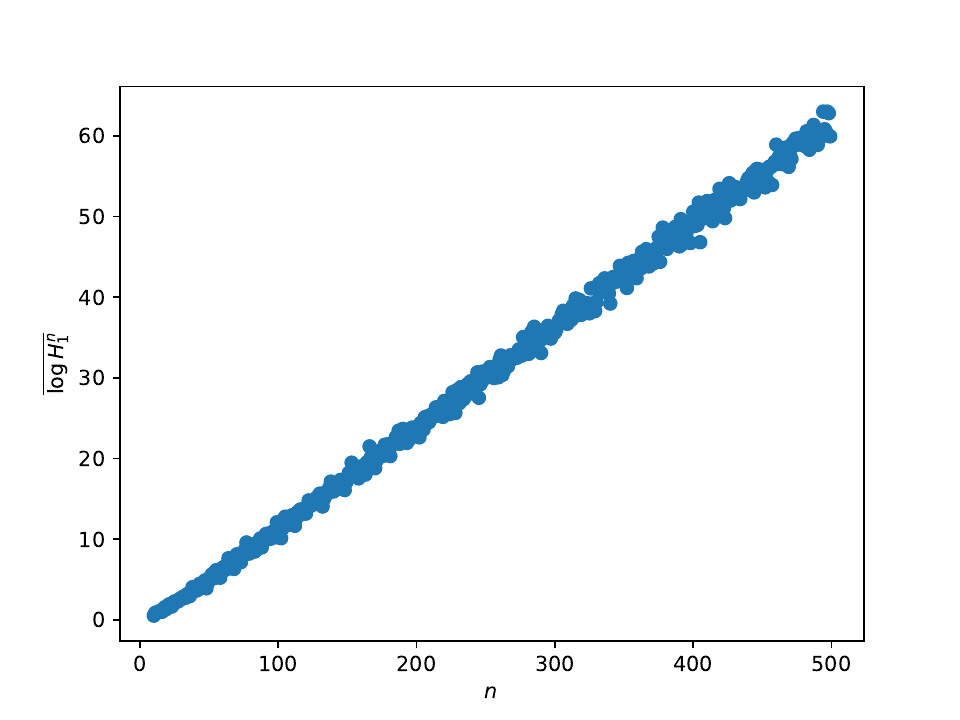}
  \caption{The estimator $\overline{\log H_1^n}$ versus $n$, computed using our method.}
  \label{fig: order genus 2} 
\end{subfigure}%
\hspace{1.5em}
\begin{subfigure}{.475\textwidth}
  \centering
  \includegraphics[width=.85\linewidth]{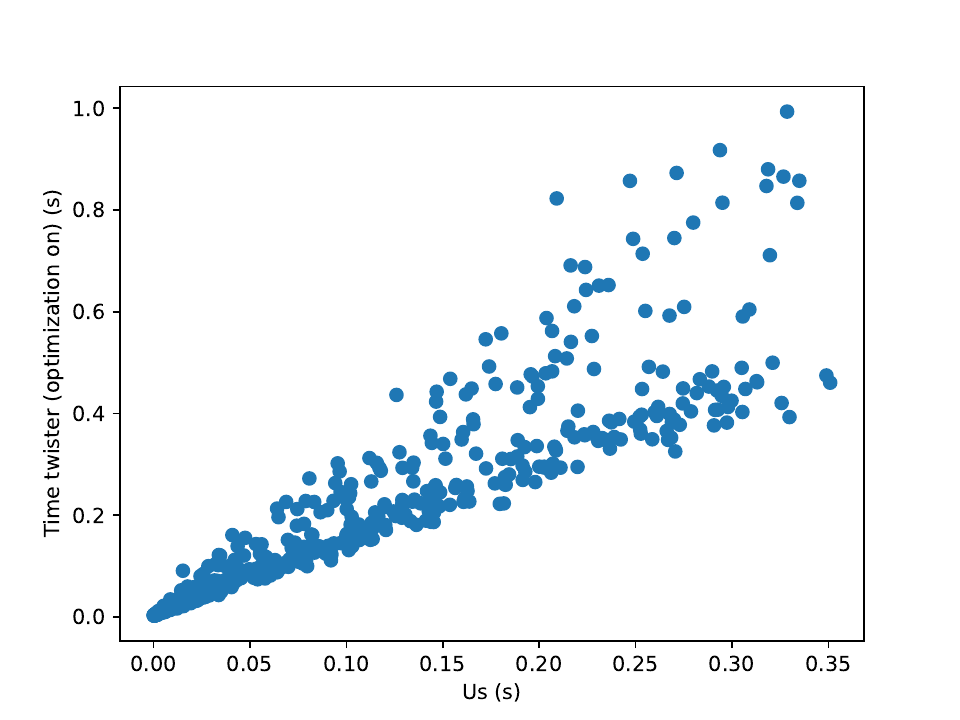}
  \caption{\texttt{Twister}'s (with optimization) average computational time of $\overline{\log H_1^n}$ versus our method's.}
  \label{fig: time genus 2}
\end{subfigure}
\caption{Experiments with computing the estimator $\overline{\log H_1^n}$ for Heegaard splittings over the double torus. Averages were taken over 30 repetitions.}
\label{fig:fig}
\end{figure}

Interestingly, although the random model is naturally suitable for empirical study, to our knowledge, there is a notable lack of works that profit from this approach, with one prominent exception~\cite{rivin2014statistics}. 
In this paper, Rivin shows that the variable $\overline{\log H_1^n}=\frac{1}{n}\sum_{i=1}^N \log|H_1((\Sigma_g,\phi_n^i(\alpha)))|$ where $\{\phi_n^i\}_{i=1}^N\overset{\mathrm{iid}}{\sim}\text{Uniform}(B(n))$ follows a Central Limit Theorem
\begin{equation}\label{eq: law genus 2}
    \frac{1}{\sqrt{N}}(\overline{\log H_1^n}-n\lambda)\to \mathcal{N}(0,\sigma^2)
\end{equation}
for some positive constants $\sigma$ and $\lambda$, as $N\to \infty$ (restricted to cases where $\log|H_1((\Sigma_g,\phi_n^i(\alpha)))|$ is finite).
In particular, he uses \texttt{Twister} to verify this relation for the case of $g=2$.
In Figure \ref{fig: order genus 2}, we reproduce the analysis and estimate $\lambda$ as 0.13. 
This time, we are about 2.5 times faster than \texttt{Twister} in the average case (\ref{fig: time genus 2}).
That both methods converge to the expected regression line is a strong indication of the correctness of the implementations.

\begin{figure}[h]
    \centering
    % Row 1
    \begin{subfigure}[b]{0.45\textwidth}
        \centering
        \includegraphics[width=\textwidth]{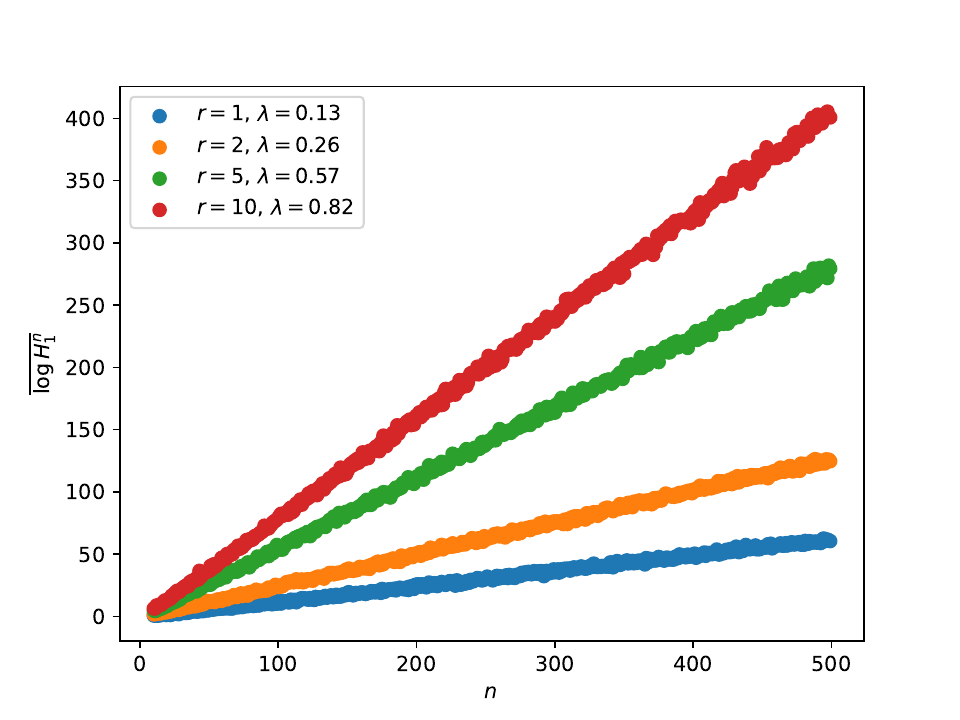}
        \caption{The estimator $\overline{\log H_1^n}$ versus $n$ where the words sampled are $\phi_{n,r}$.}
        \label{fig: different r}
    \end{subfigure}
    \hfill
    \begin{subfigure}[b]{0.45\textwidth}
        \centering
        \includegraphics[width=\textwidth]{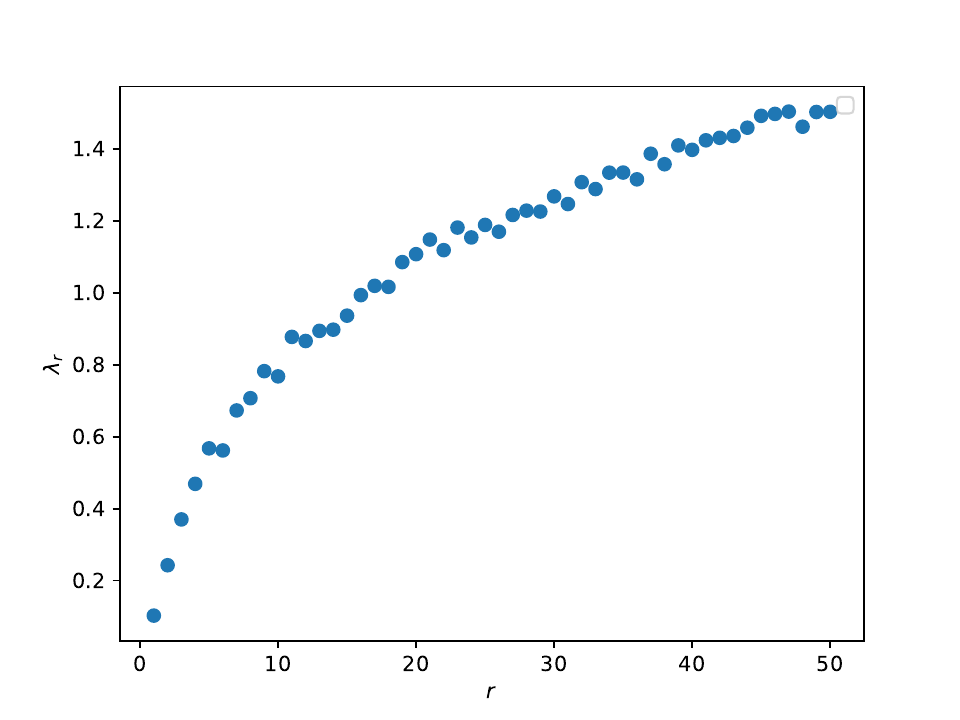}
        \caption{The estimated coefficients $\lambda$ as a function of $r$.}
        \label{fig: lambdas as a function of r}
    \end{subfigure}

    \vspace{1em} % space between rows

    % Row 2
    \begin{subfigure}[b]{0.45\textwidth}
        \centering
        \includegraphics[width=\textwidth]{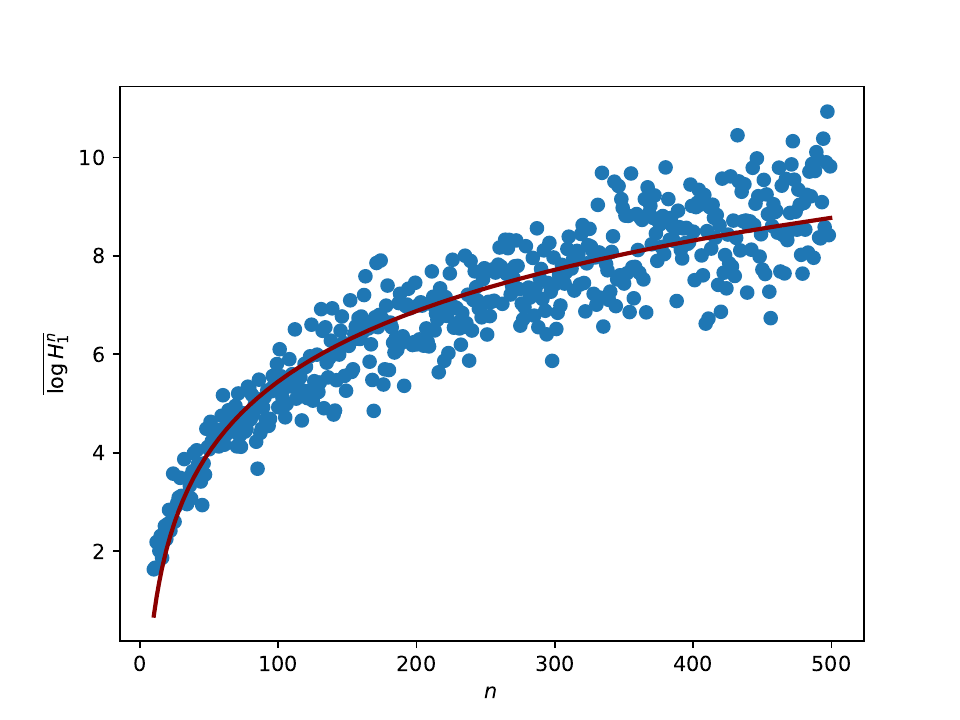}
        \caption{The estimator $\overline{\log H_1^n}$ versus $n$ where the words sampled are $\widetilde{\phi}_n$.}
        \label{fig: different r2}
    \end{subfigure}
    \hfill
    \begin{subfigure}[b]{0.45\textwidth}
        \centering
        \includegraphics[width=\textwidth]{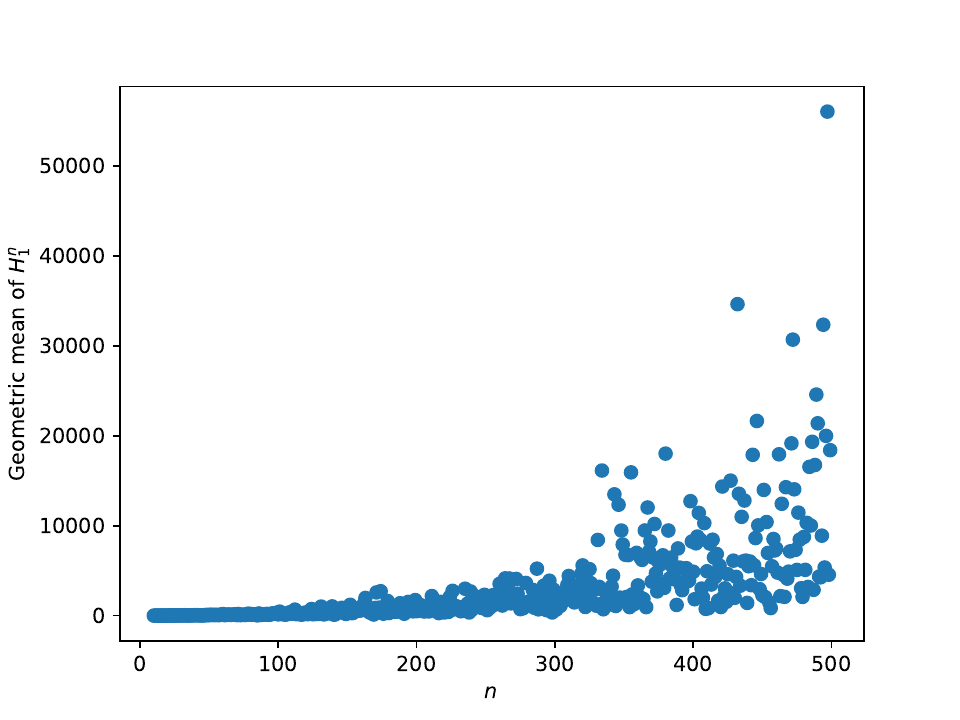}
        \caption{The estimated geometric mean of $H_1^n$ for different $n$.}
        \label{fig: geometric mean}
    \end{subfigure}

    \caption{Experiments with computing the estimator $\overline{H_1^n}$ for Heegaard splittings over the double torus with different sampling strategies. Averages were taken over 30 repetitions.}
\end{figure}

Due to the efficiency of our technique for input in power notation, we are naturally inclined to consider modifications to the sampling strategy used to compute $\overline{\log H_1^n}$ by drawing words from different distributions. 
For example, Figure \ref{fig: different r} shows the estimator $\overline{\log H_1^n}$ as a function of $n$, where we now draw words $\{\phi_{n,r}^i\}_{i=1}^N$ uniformly from $\mathcal{L}\times \llbracket -r, r\rrbracket$ for different choices of $r$. 
The distribution of $\phi_{n,1}$ is identical to $\phi_{n}$, described in the previous paragraph, and consequently, follows the linear relation described by Rivin.
However, it is noticeable that the relationship between $n$ and $\overline{\log H_1^n}$ is still linear, even for different values of $r$, although $\lambda_r$ grows as $\log(r)$ (Figure \ref{fig: lambdas as a function of r}). 
This logarithmic relationship should come as no surprise, given the direct dependency of $H_1$ on the minimal number of intersections between $\alpha$ and $\beta$ curves up to isotopies.

A more interesting modification to the sampling technique, this time denoted by $\{\widetilde{\phi}_{n}^i\}$, yields Figure \ref{fig: different r2}. 
For $\{\widetilde{\phi}_n^i\}$, we fix a value of $n$ and sample $(\tau_{s_j},k_j)$ uniformly from $\mathcal{L}\times\llbracket -n_j, n_j\rrbracket$, with the constraint that, if the word sampled equals $\tau_{s_\ell}^{k_\ell}\circ\dots\circ \tau_{s_1}^{k_1}$, then $\sum_j |k_j|=n$.
This distribution aims to mimic the random model, giving higher weights to words with more frequent repetitions. 
As indicated by the figure, $n$ and $\overline{\log H_1^n}$ seem to follow a logarithmic relation.
The estimated model is shown in red; its $R^2$ was found to be $0.86$.
Due to the logarithmic form, one is inclined to conjecture that there is a power-law relation between the geometric mean of $H_1^n$ and $n$ when sampled according to $\{\widetilde{\phi}_n^i\}$; nonetheless, as indicated by Figure \ref{fig: geometric mean}, even if such a relation exists, the distribution has heavy tails. 
We leave for future work a theoretical justification for these relations.

\section{Open questions towards a census of 3-manifolds}\label{sec: open questions} 
The strategy for generating splittings used in Section \ref{sec: examples and experiments} suggests a natural construction of a census of closed 3-manifolds containing elements that would otherwise be inaccessible to triangulation-based approaches. 
Being more explicit about the ``recipe'' presented in the introduction, one could
\begin{enumerate}
    \item~\label{item: postcensus1} Fix a genus $g$ and generate Heegaard words in power-notation form up to a maximum length, which are then used to compute diagrams;
    \item~\label{item: postcensus2} Greedly apply combinatorial moves (such as \texttt{GET EFFICIENT POSITION}, \texttt{DO A DISK SLIDE}, \texttt{DO DESTABILIZATION}, and \texttt{DETECT REDUCTION}) to simplify as much as possible the diagrams and detect some trivial, previously identified instances, such as stabilizations; 
    \item~\label{item: postcensus3} Use 3-manifold invariants to detect redundancies in the enumeration of spaces.
\end{enumerate}

A few points about these strategies deserve mention.
Instead of sampling Heegaard words as in item~\ref{item: postcensus1}, a more natural (and less studied) model would come from directly sampling the diagrams as normal coordinates, without regard for an input word.
Although normal coordinates respect some well-defined combinatorial constraints, efficiently restricting the multicurves to systems in the surface (without, for example, appealing to Theorem \ref{th: check diagram}) does not seem to be an easy task. 
Similar problems appear when sampling 3-manifolds from random triangulations, for which some asymptotic results are already known \cite{petri2022model}.

Now for item~\ref{item: postcensus3}, we note that the very efficient computations of (presentations of) the fundamental and homology groups in our data structure represent some advances, but the difficulty of the isomorphism problem for the first~\cite{aschenbrenner20153} and the lack of strong distinguishability of the second make them not ideal in most cases. 
Alternatively, the authors of~\cite{costantino2020kuperberg} suggest a presentation of the much distinguishable Kuperberg~\cite{kuperberg1991involutory} and Turaev-Viro~\cite{turaev2010quantum} invariants computed directly from Heegaard diagrams of 3-manifolds, but an algorithm on our data structure and its complexity for the cases of interest still need to be determined. 

Finally, in addition to the exploration of exponentially more spaces, a Heegaard splitting census can probably be used to address the well-reported difficulty of generating 3-manifolds with particular topological structures---such as hyperbolicity---observed in triangulation-based methods \cite{matveev2007algorithmic}.
This is because much of the structure of a 3-manifold is encoded by the \emph{Hempel distance} of its splittings, defined as a graph distance in the infinite complex of meridians of the involved handlebodies, loosely related to the number of intersections between $\alpha$ and $\beta$ \cite{hempel20013}. 
For example, manifolds that have splittings with a Hempel distance of at least 1 are irreducible \cite{kobayashi1988casson}, distance of at least 3 are hyperbolic \cite{hempel20013}, and distance of at least $k$ cannot embed incompressible orientable surfaces of the genus at most $2k$ \cite{hartshorn2002heegaard}. 
Bounds relating the hyperbolic volume \cite{hamenstadt2022stability,kobayashi2013hyperbolicvolumeheegaarddistance} and the Heegaard genus \cite{scharlemann2006alternate} with the Hempel distance are also known. 
There exists a well-known algorithm, due to Masur and Schleimer \cite{masur2013geometry}, to approximate the Hempel distance from a Heegaard diagram, which can probably be translated into our structure, but whose complexity is not clear. 
Alternatively, the authors of \cite{ennes2025hardnesscomputationquantuminvariants} give a polynomial-time algorithm that uses (uncompressed) intersection sequences of Heegaard splittings to increase the Hempel distance while keeping some quantum invariant fixed throughout the process, making it therefore possible to automatically generate sets of hyperbolic manifolds that are guaranteed not to be homeomorphic.

\bibliographystyle{plain}
\bibliography{compressed}

\appendix

\section{Extension to general systems}\label{app: extensions}
As mentioned in Remark \ref{rm: systems}, there is a generalization that allows systems to have more than $g$ components. 
In particular, for this section, we will assume that a \emph{system} is a multicurve of cardinality at least $g$ and with no isotopic components.
A standard Euler characteristic argument \cite{schleimer2006notes} implies that a system $\gamma$ cuts $\Sigma_g$ into $\#\gamma-g+1$ at least thrice punctured spheres.
When $\#\gamma=g$, we say that the system is \emph{minimal} and it reduces to the definition used in the main text. 
If, on the other hand, $\#\gamma = 3g-3$, we say that the system is \emph{maximal} or a \emph{pants decomposition} because it cuts $\Sigma_g$ into a union of three-times punctured spheres known as \emph{pairs of pants}. 

Like minimal systems, general systems uniquely define handlebodies. 
This means that the notion of equivalence naturally extends to the more general systems, even if they have different cardinality. 
In particular, one \cite{ennes2025hardnesscomputationquantuminvariants, johnson2006notes} can show that, provided that $\gamma$ is not a pants decomposition, the band sum $c$ of two punctures $c_1$ and $c_2$ and an arc $p$ connecting them in $\Sigma_g\backslash \gamma$ is a meridian of $\mathcal{H}_\gamma$ and is not isotopic to any other component of $\gamma$. 
This means that $\gamma'=\gamma\cup\{c\}$ forms a system which is equivalent to $\gamma$; we call $\gamma'$ an \emph{extension} of $\gamma$.

We can expand our data structure to allow for more general systems for $\alpha$ and $\beta$. 
In particular, a diagram $(\Sigma_g,\alpha,\beta)$ is \emph{minimal} if $\alpha$ and $\beta$ are minimal systems; similarly, it is \emph{maximal} if they are pants decompositions. 
In fact, although more resource-intensive, this general version of Heegaard diagrams has been used both in mathematical (for example, \cite{hempel20013, yoshizawa2014high}) and in computer science (\cite{ennes2025hardnesscomputationquantuminvariants}) settings.
All algorithms described in this paper can be adapted to this more general notion of Heegaard diagrams with only minor modifications.
We can, moreover, describe efficient methods to translate between minimal and more general systems.
This is exactly the content of the next two results. 

\begin{theorem}[\texttt{REDUCE TO MINIMAL}]
    Let $(T,E_T(\alpha),N_T(\beta))$ be a Heegaard diagram of complexity $m$. Then one can compute an equivalent minimal Heegaard diagram $(T,E_T(\alpha'),N_T(\beta'))$ in time $O(m|T|)$.
\end{theorem}
\begin{proof}
    Compute $S(T,\beta)$ and delete the separating components of $\beta$ by using breadth-first search in the dual graph. 
    In total, this takes time $O(m|T|+g|T|)=O(m|T|)$ since $O(g)=O(|T|)=O(\sqrt{m})$. 
    Similarly, delete any separating curve of $\alpha$ in time $O(g|T|)$. 
\end{proof}

\begin{theorem}[\texttt{EXTEND TO MAXIMAL}]\label{th: extend to maximal}
    Let $(T',E_{T'}(\alpha),N_{T;}(\beta))$ be a Heegaard diagram of complexity $m$. Then one can compute an equivalent maximal Heegaard diagram $(T,E_T(\alpha'),N_T(\beta'))$ in time $O(gm|T|)$.
\end{theorem}
\begin{proof}
    Without loss of generality, suppose $(T,E_T(\alpha),E_T(\beta))$ is a minimal diagram; we want to extend each multicurve to a pants decomposition. 
    Both cases can be treated similarly thanks to the following lemma.

\begin{lemma}\label{lm: get pants}
    Suppose $\gamma$ is an edged minimal system in a cellular embedding $T$ of a surface $\Sigma_g$. Then one can find a pants decomposition $\gamma'$, equivalent to $\gamma$, in time $O(g|T|)$.
\end{lemma}
\begin{proof}
    Note that $\Sigma_g\backslash \gamma$ is homeomorphic to the $2g$ punctured sphere, so let $d_1,\dots,d_{2g}$ represent the boundary components of this surface.
    We will find some $2g-3$ new separating curves $d_{2g+1},\dots, d_{3g-3}$ in $\Sigma_g$ that recursively define pairs of pants as follows.
    In step $i$, choose a curve $d_{g+j}, 0\leq j\leq i-1$ from a list of available curves. 
    Call this choice $s_1$ and pick a face $t$ adjacent to $s_1$. 
    Using breadth-first search, we find, in time $O(|T|)$, the shortest path $p$ in the dual graph starting from $t$ to any face adjacent to one of the available curves.
    Call the closest available curve $s_2$ and remove $s_1$ and $s_2$ from the list of available curves. 
    Note that we can encode $p$ as a list of faces of $T$.
    In time $O(\|E_T(s_1)\|)=O(|T|)$, we update $p$ by iteratively removing all faces adjacent to $s_1$ except for the last one.
    By construction, $p$ intersects each edge of the cellular complex at most once and does not intersect any face adjacent to an available curve, except for exactly one face adjacent to $E_T({s_1})$ and one face adjacent to $E_T({s_2})$.
    
    As in the proof of Theorem \ref{th: do a disk slide}, make $I_T(d_i)=I_T(s_2)^{-1}\circ I_T(p)^{-1}\circ I_T(s_1)\circ I_T(p)$, where $p$ is potentially trivial, and add $d_i$ to the list of available curves. 
    The curve $d_i$ separates a thrice-punctured sphere component with punctures defined by $s_1$, $s_2$, and $d_i$, from a $2g-i$ punctured sphere component, and which component is which can be identified in time $O(|T|)$.
    We then update $T$ to a cellular embedding of $\Sigma_0^{2g+i}$ by adding the edges and vertices defined by $I_T(d_i)$ to the graph and blocking paths in the dual graph transverse to them.
    While $T$ becomes more complicated at each step, because the arc $p$ is constructed to intersect a constant number of faces adjacent to the available curves and the curves $d_{g+j}$, $1\leq j\leq i-1$, are separating, at each step, breadth-first needs only to consider $O(|T|)$ many edges and vertices of the dual graph of $T$. 
    The algorithm terminates when $i=3g-3$, in which case we have defined, by construction, a maximal system $\gamma = \{c_1,\dots,c_g,d_{g+1},\dots,d_{3g-3}\}$.
\end{proof}

We now apply Lemma \ref{lm: get pants} to get some pants decomposition equivalent to $\alpha$, edged in a new triangulation $T'$. 
Because $|T'|=O(|T|)$, we can use the discussion of corner coordinates to find the coordinates $N_{T'}(\beta)$ in time $O(g|T|)$. 
We then compute some edged pants decomposition equivalent to $\beta$ in $S(T',\beta)$. 
Finally, by the very proof of Lemma \ref{lm: get pants} and part (b) of Theorem \ref{th: erickson}, this gives the normal coordinates of the new $\beta$ curves with respect to $T'$.
\end{proof}

\begin{remark}
    Theorem \ref{th: extend to maximal} \emph{does not} guarantee that there exists a $\phi\in\Mod(\Sigma_g)$ such that $\beta'=\phi(\alpha')$ where $\alpha'$ and $\beta'$ are pants decompositions.
    In fact, it is a well-known result from surface theory \cite{moser2014upper} that such a homeomorphism exists if and only if the \emph{pants decomposition graphs} associated with $\alpha'$ and $\beta'$ are isotopic.
    One can easily imagine a modification to the algorithm of the theorem that constructs a pants decomposition $\beta'$ that defines a graph isotopic to the graph defined by $\alpha'$ by forcing arcs between particular punctures of $\Sigma_g\backslash \alpha'$ as in Theorem \ref{th: do a disk slide}, but that would imply an algorithm exponential on $|T|$. Alternatively, a method based on Hatcher and Thurston's proof of connectivity of the \emph{pants graph} (not to be confused with the pants decomposition graph)~\cite{hatcher1980presentation}[Appendix] can be devised, but efficiency in $g$ is not trivial.
\end{remark}

\section{Details on the implementations of Section \ref{sec: examples and experiments}}\label{app: details}

\begin{figure}
    \centering
    \includegraphics[width=0.65\textwidth]{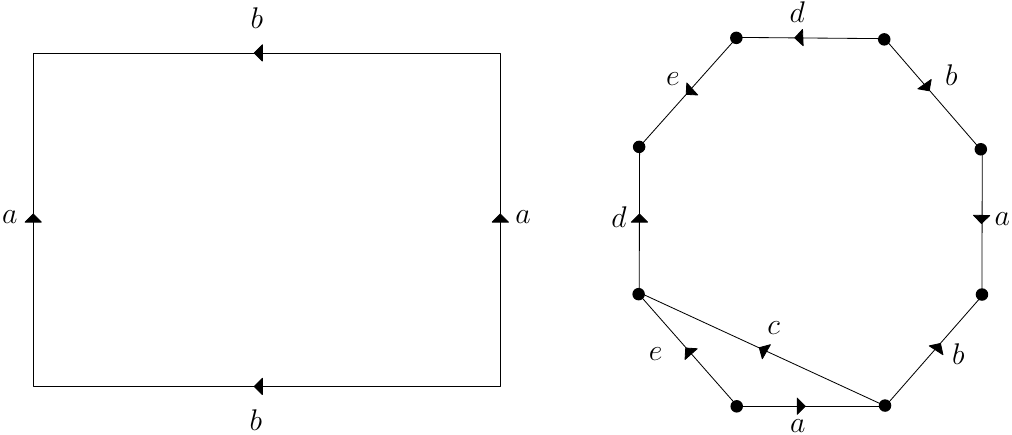}
     \caption{The cellular embeddings $T_g$ for $g=1$ (left) and $g=2$ (right) with the choices of edge labels agreeing with \texttt{Twister}.}
  \label{fig: Tg}
\end{figure}

As mentioned in the main text, several simplifying assumptions were made for the implementation of the examples in Section \ref{sec: examples and experiments} and Appendix \ref{sec: genus 2}. 
First and most importantly, we \emph{did not} compute a data structure for the Heegaard diagrams as described in Section \ref{sec: data structure}. 
Instead, we use a general cellular embedding $T_g$ given by the fundamental polygons of $\Sigma_g$ \cite{hatcher2002algebraic} to represent the surfaces, which implies that, instead of normal coordinates, we assume the $\beta$ curves to be given by SLPs $I_T(\beta)$.
Although this deviates from what was described in this paper, we note that cyclically reducing the SLPs is irrelevant to finding $H_1(M)$ (or $\pi_1(M)$) and introduces some unbearable computational cost.
We also assumed the Lickorish generators to be edged with respect to $T_g$ and that their intersection sequences are computed only once, say as a preprocessing step, using Lemma \ref{lm: get twins} (in fact, for these experiments, we computed the SLPs by hand).
Figure \ref{fig: Tg} shows $T_g$, for $g=1$ and $2$, while Table \ref{tb: implementation} indicates the edges and (uncompressed) intersection sequences of the curves involved. 

\begin{table}[h!]
\centering
\begin{tabular}{ |c|c|c|c|c| } 
\hline
Genus & $E_{T_g}(\alpha)$ & $I_{T_g}(\alpha)$ &$E_{T_g}(\mathcal{L})$ & $I_{T_g}(\mathcal{L})$ \\[0.6 ex]
\hline

 &&&&\\
$g=1$ & $[ a]$ & $ b$
& $[ b]$ & $a^{-1}$\\
 &&&&\\
\hline
 &&&&\\
 & \multirow{3}{*}{{\shortstack{$[ a ]$\\\vspace{1ex}$[ e ]$}}}  & \multirow{3}{*}{{\shortstack{$b$ \\\vspace{1ex}$d^{-1}$}}} &\multirow{3}{*}{\shortstack{$[ b ]$\\$[ c ]$\\$[ d ]$}} & \multirow{3}{*}{\shortstack{$c^{-1}\cdot a $\\$ b\cdot d $\\$ c^{-1}\cdot e^{-1} $}}\\ 
$g=2$&  &  & &\\ 
&  &  & &\\ 
 &&&&\\
\hline
\end{tabular}
\caption{Here, we use $[ \cdot ]$ to denote lists. We assume that the $\alpha$ curves are contained in the Lickorish generators, meaning that they are omitted in the fourth and fifth columns.}
\label{tb: implementation}
\end{table}

\end{document}